%% file: RTOHierarchical_a.tex
\documentclass{siamart190516}

\usepackage[tmargin=0.8in, bmargin=0.9in, centering]{geometry}
\usepackage{graphicx}
\usepackage{subcaption}
\usepackage{enumitem}
\usepackage{amsmath,amssymb}
\usepackage{booktabs}
\usepackage{mathtools} 				%
\usepackage{algorithm}
\usepackage{algpseudocode}
\usepackage[mathscr]{euscript}
\usepackage[normalem]{ulem} 	%
\usepackage{pbox} 						%
\usepackage{comment}
\usepackage{todonotes}
\usepackage{subdepth}
\usepackage{dsfont}

\newcommand{\pic}[2]{\includegraphics[width = #1\textwidth]{#2}}

\def\bm#1{\mbox{\boldmath$#1$}}
\DeclareMathOperator*{\argmin}{arg\,min}

\newcommand{\bb}{{\bf b}}

\newcommand{\mb}{{\bf m}}
\newcommand{\fb}{{\bf f}}
\newcommand{\gb}{{\bf g}}

\newcommand{\ub}{{\bf u}}
\newcommand{\vb}{{\bf v}}
\newcommand{\wb}{{\bf w}}
\newcommand{\xb}{{\bf x}}
\newcommand{\yb}{{\bf y}}
\newcommand{\zb}{{\bf z}}
\newcommand{\Ab}{{\bf A}}
\newcommand{\Bb}{{\bf B}}
\newcommand{\Cb}{{\bf C}}
\newcommand{\Db}{{\bf D}}

\newcommand{\Fb}{{\bf F}}
\newcommand{\Gb}{{\bf G}}
\newcommand{\Ib}{{\bf I}}
\newcommand{\Jb}{{\bf J}}

\newcommand{\Mb}{{\bf M}}
\newcommand{\Kb}{{\bf K}}
\newcommand{\Hb}{{\bf H}}
\newcommand{\Pb}{{\bf P}}
\newcommand{\Qb}{{\bf Q}}

\newcommand{\Sb}{{\bf S}}
\newcommand{\Tb}{{\bf T}}
\newcommand{\Ub}{{\bf U}}

\newcommand{\Yb}{{\bf Y}}
\newcommand{\Xb}{{\bf X}}
\newcommand{\zerob}{{\bf 0}}

\newcommand{\etab}{\bm{\eta}}

\newcommand{\xib}{\bm{\xi}}
\newcommand{\zetab}{\bm{\zeta}}

\newcommand{\Phib}{\bm{\Phi}}
\newcommand{\Psib}{\bm{\Psi}}
\newcommand{\Pib}{\bm{\Pi}}
\newcommand{\Sigmab}{\bm{\Sigma}}

\newcommand{\thetab}{\bm{\theta}}
\newcommand{\N}{\mathcal{N}}

\newcommand{\R}{\mathbb{R}}

\newcommand{\param}{\ub}
\newcommand{\data}{\yb}
\newcommand{\predata}{\boldsymbol\eta}

\newcommand{\obscov}{\Sigmab}
\newcommand{\prmean}{\mb}

\newcommand{\prcovh}{{\Cb}_{\gamma}}
\newcommand{\prprech}{\Pb_\gamma}
\newcommand{\forward}{\Fb}

\newcommand{\like}{\mathcal{L}}

\newtheorem{assume}[theorem]{Assumption}
\newsiamremark{remark}{Remark}

\headers{Optimization-Based MCMC for Hierarchical Bayesian Inference}{J. Bardsley, T. Cui}

\title{Optimization-Based MCMC Methods for Nonlinear Hierarchical Statistical Inverse Problems}

\author{
  Johnathan M.~Bardsley%
  \thanks{Department of Mathematical Sciences, University of Montana, Missoula, MT 59812 USA (\email{bardsleyj@mso.umt.edu}).}
  \and
  Tiangang Cui%
  \thanks{School of Mathematics, Monash University, Victoria 3800, Australia (\email{tiangang.cui@monash.edu}).}
}

\begin{document}

\setlength{\abovedisplayskip}{5pt}
\setlength{\belowdisplayskip}{5pt}
\setlength{\belowcaptionskip}{0pt}
\setlength{\textfloatsep}{6pt}
\setlength{\intextsep}{6pt}

\maketitle

\begin{abstract}
In many hierarchical inverse problems, not only do we want to estimate high- or infinite-dimensional model parameters in the parameter-to-observable maps, but we also have to estimate hyperparameters that represent critical assumptions in the statistical and mathematical modeling processes. As a joint effect of high-dimensionality, nonlinear dependence, and non-concave structures in the joint posterior posterior distribution over model parameters and hyperparameters, solving inverse problems in the hierarchical Bayesian setting poses a significant computational challenge. In this work, we aim to develop scalable optimization-based Markov chain Monte Carlo (MCMC) methods for solving hierarchical Bayesian inverse problems with nonlinear parameter-to-observable maps and a broader class of hyperparameters.
Our algorithmic development is based on the recently developed scalable randomize-then-optimize (RTO) method \cite{wang2019scalable} for exploring the high- or infinite-dimensional model parameter space.
By using RTO either as a proposal distribution in a Metropolis-within-Gibbs update or as a biasing distribution in the pseudo-marginal MCMC \cite{andrieu2009pseudo}, we are able to design efficient sampling tools for hierarchical Bayesian inversion. 
In particular, the integration of RTO and the pseudo-marginal MCMC has sampling performance robust to model parameter dimensions. 
We also extend our methods to nonlinear inverse problems with Poisson-distributed measurements. 
Numerical examples in PDE-constrained inverse problems and positron emission tomography (PET) are used to demonstrate the performance of our methods. 
\end{abstract}



\input{sec_intro_back_a}
\input{sec_RTO_a}

\input{sec_Gibbs_a}

\input{sec_PM_a}

\input{sec_Numerics_Conclusion_a}

\section*{Acknowledgments}
We thank J. Heikkinen for providing us with the code used in setting up Example 2 and J. Rieger for fruitful discussions on the trust-region modification of RTO.
J.~Bardsley acknowledges support from the Gordon Preston Fellowship offered by Monash University.
T.~Cui acknowledges support from the Australian Research Council, under grant number CE140100049.
\appendix

\input{sec_appendix_a}

\bibliographystyle{siamplain}
\bibliography{references}

\end{document}

%% file: sec_intro_back_a.tex
\section{Introduction}
\label{sec:Intro}

At the heart of many mathematical modelling problems, there often lies an inverse problem that aims to estimate unknown parameters of mathematical models from noisy and indirect observations.
Due to smoothing properties of the parameter-to-observable map and incompleteness of data, such inverse problems are often ill-posed: there may exist many feasible parameters that are consistent with the observed data, and small perturbations in the data may lead to large perturbations in unregularized parameter estimates.
To remove the ill-posedness, the Bayesian approach \cite{IP:KaiSo_2005, IP:Tarantola_2004, IP:Stuart_2010} casts the solution of inverse problems as the posterior probability distribution of the model parameters conditioned on the data.

In a typical Bayesian inverse problem, unknown model parameters are often represented as functions, and thus yield high-dimensional discretized representations. This way, exploring the high-dimensional posterior distribution is in general a computationally challenging task.
Recently, many efficient methods have been developed to tackle this challenge; for example, (preconditioned) Crank-Nicolson (pCN) methods \cite{MCMC:BRSV_2008, MCMC:CRSW_2013} that establish the foundation for designing and analysing MCMC algorithms in a function space setting, stochastic Newton methods \cite{MCMC:MWBG_2012,MCMC:Petra_etal_2014} that utilise Hessian information to accelerate the convergence, operator-weighted methods \cite{MCMC:CLM_2016,MCMC:Law_2014,MCMC:DS_2018} that generalise PCN methods using (potentially location-dependent) operators to adapt to the geometry of the posterior, as well as optimization-based sampling methods \cite{wang2019scalable,Sto:BSHL_2014,Sto:MTAC_2012,oliver2017metropolized,wang2018randomized} that convert scalable optimization algorithms into MCMC samplers.

In addition to the high-dimensional model parameters, we often need to introduce {\em hyperparameters} to describe various modelling assumptions in an inverse problem.
Such hyperparameters can be used to characterise the statistical model of the observational noise, e.g., its variance, and to describe the variation and correlation structure of the prior distribution of model parameters.
See \cite{agapiou2014analysis,bardsley2010hierarchical,calvetti2008hypermodels,calvetti2009conditionally,dunlop2017hierarchical,fox2016fast,lucka2012fast,wang2004hierarchical} and references therein for further details.
In many physical applications, we may also use hyperparameters to parametrize assumptions in the parameter-to-observable map, for instance, the relative permeability curves in subsurface modeling (e.g., \cite{IP:CFO_2011,cui2019posteriori}) and the intensity of radiation sources in the positron emission tomography (see Section \ref{sec:PET} for details).
In this setting, we need to characterize the joint posterior distribution of model parameters and hyperparameters conditioned on the observed data, which is often referred to as the {\em hierarchical Bayesian inference}.
Here we aim to design optimization-based sampling methods that can explore the joint posterior distribution for nonlinear inverse problems and can handle a broader class of hyperparameters.

Many of the existing efficient posterior exploration methods focus on accelerating the posterior sampling for a fixed set of hyperparameters.
Since the hyperparameters and model parameters often have complicated and non-concave interactions, significant extensions to the existing works are needed to obtain efficient samplers to explore the joint posterior distribution.
For linear inverse problems, \cite{bardsley2012mcmc} investigated the use of Gibbs sampling schemes that alternatively update the model parameters and hyperparameters, \cite{agapiou2014analysis} analyzed the dimension scalability (w.r.t. the model parameters) of several Gibbs sampling schemes, \cite{dunlop2019hyperparameter} analyzed the consistency of the hyperparameter estimation, \cite{fox2016fast,saibaba2019} investigated the use of the one-block-update of \cite{rue2005gaussian} and marginalization over model parameters to accelerate the sampling.
The success of these developments commonly relies on two facts: there exists an analytic expression for the marginal posterior over the hyperparameters and one can the directly sample the conditional posterior over the model parameters for given hyperparameters.
However, these are no longer the case for nonlinear inverse problems.

In this work, we will present several new MCMC methods for sampling the joint posterior distribution for nonlinear inverse problems.
We will also consider broader classes of likelihood functions and prior distributions. This includes Poisson observation processes that arise in PET imaging and in the estimation of unknown correlation structures in the prior distribution.
Our algorithmic development is based on non-trivial extensions of the {\em randomize-then-optimize} (RTO) method \cite{Sto:BSHL_2014}.
As detailed in Section \ref{sec:RTO}, we will first present an efficient implementation of RTO that takes advantage of intrinsic low rank structures of inverse problems, and then discuss several theoretical properties and generalizations to make RTO suitable for the hierarchical setting and the Poisson likelihood.
Then, we integrate RTO into the Metropolis-within-Gibbs method to present computationally efficient strategies to alternatively update model parameters and hyperparameters.
The resulting RTO-within-Gibbs sampler shares similar dimension scalability properties of the centred Gibbs scheme of \cite {roberts2001inference,yu2011center} in the linear setting \cite {agapiou2014analysis}, and thus can deteriorate with the model parameter dimension.
To overcome this difficulty, we will also combine RTO with the pseudo-marginal (PM) principle \cite {andrieu2009pseudo}  to design MCMC methods that are robust with parameter dimensions.

This paper is organized as follows. In Section \ref{sec:back}, we discuss hierarchical Bayesian inverse problems. In Section \ref{sec:RTO}, we present an efficient implementation of RTO and its generalizations. In Sections \ref{sec:Gibbs_algorithms} and \ref{sec:Pseudo_algorithms}, we present MCMC algorithms for sampling the joint posterior distributions . We present numerical experiments in Sections \ref{sec:elliptic} and \ref{sec:PET}, and end with discussions in Section \ref{sec:Conclusions}.

\section{Hierarchical Bayesian inverse problems}
\label{sec:back}
Here we will define the prior distributions, likelihood functions, and various forms and elements of the posterior distributions in our hierarchical inverse problems.
Throughout this paper, given a positive definite matrix $\Ab$, we denote the matrix weighted inner product by $\langle \ub , \vb  \rangle_{\Ab} = \langle \ub , \Ab\vb  \rangle$ and let $\Vert \ub \Vert_{\Ab} = \sqrt{\langle \ub , \ub  \rangle_{\Ab}}$ be the induced norm.

\subsection{Prior modelling}
\label{sec:prior}
In an inverse problem, we seek to infer the unknown parameters of a mathematical model from observed data that correspond to the observable model outputs.
The unknown parameter $u(s), s \in \Omega$ is some heterogeneous function belong to a separable Hilbert space $\mathcal{H}(\Omega)$ for a given domain $\Omega$.
We begin by introducing the Gaussian process prior $\mu_0 = \N(m, \delta^{-1}\mathcal{C}_\gamma)$, where $m$ is the mean function and $\delta^{-1}\mathcal{C}_\gamma$ is the covariance operator, to represent the {\it a priori} information about the parameter. Here $\delta\in \R_{>0}$ is the precision parameter that controls the variance of the Gaussian process.
We parametrize the covariance operator by the hyperparameter $\gamma$ to account for possible changes in the correlation structure.
For a given $\gamma$, the covariance $\mathcal{C}_\gamma$ should be a symmetric, positive, and trace-class operator such that $\mu_0(\mathcal{H}) = 1$.
This way, we can represent the parameter $u(s)$ and prior covariance using a discretized grid, and the prior may yield an infinite dimensional limit under grid refinement (see \cite{IP:BGMS_2013,IP:Stuart_2010}).

Suppose the parameter function $u(s)$ is evaluated on set of $n$ nodes $s_1, \ldots, s_n$ in the discretized representation. We need to operate with the discretized covariance operator and its factorizations to compute the prior density and to generate realizations from the prior.
Discretizing the covariance operator yields a covariance matrix $\prcovh$, in which each element of $\prcovh$ can be defined by a covariance function $\rho : \Omega \times \Omega \mapsto \R_{\geq0}$. For example, the widely used Mat\'{e}rn covariance function takes the form
\begin{gather}\label{eq:matern}
\rho_\nu(s_1, s_2; \gamma) =  \frac{2^{1-\nu}}{\Gamma(\nu)} \big( \gamma \, \sqrt{2\nu} \, |s_1- s_2| \big)^{\nu} \mathcal{K}_{\nu}\big(\gamma\,\sqrt{2\nu}\,|s_1- s_2|\big),
\end{gather}
where $\Gamma(\cdot)$ is the Gamma function, $\mathcal{K}_{\nu}(\cdot)$ is the modified Bessel function of the second kind, $\nu \geq \frac12$ defines the smoothness of the random process, and $\gamma$ defines the correlation length.
For $\nu = \frac12$, the Mat\'{e}rn covariance function can be simplified to the exponential covariance function
\begin{gather*}
\rho(s_1, s_2; \gamma) =  \exp\big( - \gamma\,|s_1- s_2|\big).
\end{gather*}
Note that the spatial correlation decreases with increasing $\gamma$.

The covariance matrix $\prcovh$ can be dense, and thus it can be computationally costly to directly compute its matrix vector product (which costs $\mathcal{O}(n^2)$ operations) and factorisations (which costs $\mathcal{O}(n^3)$ operations).
Many computationally efficient ways have been proposed to handle operations with the covariance matrix and its factorisations by utilizing specific structures of the covariance matrix.
For example, Karhunen-Lo\'{e}ve expansion \cite {SuMo:MarNa_2009,schwab2006karhunen} is a widely used approach that constructs a reduced approximate representation of the prior covariance via the truncated eigendecomposition of the covariance function.
For problems with a stationary covariance function and discretized on a regular grid, the circulant embedding method \cite{chan1997algorithm,graham2018analysis,wood1994simulation} employs fast Fourier transform methods to operate with the covariance matrix and its factorisations in the frequency domain in $\mathcal{O}(n\,\log(n))$ operations.
Recent investigations \cite{feischl2018fast,harbrecht2015efficient,khoromskij2009application} employ the hierarchical matrix method to approximate the covariance matrix and its factorisations, which cost $\mathcal{O}(n\,\log(n))$ operations and can be generalised to non-stationary covariance functions and general node sets.

In this work, we specify the Gaussian process prior using a Laplace-like stochastic partial differential equation (SPDE, see \cite{MRF:LRL_2011} and references therein), which takes the form
\begin{gather}\label{eq:gmrf}
\big( \gamma  - \triangle\big)^{\beta/2} u(s) = \mathcal{W}(s), \quad \text{for} \quad s\in\Omega\subset \R^d,
\end{gather}
where $\mathcal{W}(s)$ is a spatial Gaussian white noise with unit variance, $\triangle$ is the Laplace operator, and $\gamma \in \R_{>0}$ is a scalar variable used to model the correlation length of the Gaussian process.
The order of the differential operator should be sufficiently high, i.e., $\beta > d/2$ , such that the resulting covariance operator will be trace-class in $\mathcal{H}(\Omega)$.
In an infinite domain, the SPDE in \eqref{eq:gmrf} effectively defines a Gaussian process with the Mat\'{e}rn covariance function \eqref{eq:matern} with $\nu = \beta - d/2$.
We choose $\beta=1$ for $d = 1$ and $\beta = 2$ for $d = 2,3$ to satisfy this condition.

Since $\beta$ is integer-valued here, finite element methods can be employed to discretize the covariance operator defined by \eqref{eq:gmrf} and the parameter. Given a set of locally compact basis functions $\{\phi_j(s)\}_{j = 1}^{n}$, the parameter yields the finite dimensional approximation
$u(s) = \sum_{j = 1}^n \phi_j(s) u_j$.
This way, one can express the parameter function using the coefficients associated with the basis functions. This leads to the discretized parameters $\param = (u_1, u_2, \ldots, u_n)^\top$.
Similarly, we can express the mean function $m(s)$ by discretized coefficients $\prmean = (m_1, m_2, \ldots, m_n)^\top$ associated with the basis functions $\{\phi_j(s)\}_{j = 1}^{n}$.
Then, we follow the procedure in \cite{IP:BGMS_2013,MRF:LRL_2011} to formulate the covariance matrix for the discretized parameter.
Employing the Galerkin formulation to discretize the SPDE in \eqref{eq:gmrf}, we obtain the matrices $\Mb, \Kb \in \R^{n \times n}$, where each entry of $\Mb$ and $\Kb$ are specified by
\begin{gather*}
\Mb_{ij} = \langle\phi_i, \phi_j \rangle, \quad \text{and}\quad \Kb_{ij} = \langle\nabla\phi_i, \nabla\phi_j \rangle.
\end{gather*}
Since the basis functions $\{\phi_j(s)\}_{j = 1}^{n}$ are locally compact, both $\Mb$ and $\Kb$ are sparse.
We apply mass lumping to the matrix $\Mb$ to obtain a diagonal matrix $\bar\Mb$.
The discretization of \eqref{eq:gmrf} specifies the covariance matrix through its inverse $\prprech^{} \coloneqq \prcovh^{-1}$, which is known as the precision matrix.

\begin{definition}{Prior precision matrices.}\label{def:prior}
We choose $\beta = 1$ for the case $d=1$. This way, the discretized precision matrix takes the form
\begin{gather}
\prprech = (\gamma \,\bar\Mb + \Kb ),
\end{gather}
We choose $\beta = 2$ for the cases $d=2, 3$, which yields the discretized precision matrix
\begin{gather}
\prprech = (\gamma \,\bar\Mb + \Kb )\,\bar\Mb^{-1}\,(\gamma \,\bar\Mb + \Kb ) = \gamma^2 \,\bar\Mb + 2 \gamma\,\Kb + \Kb\,\bar\Mb^{-1}\Kb  .
\end{gather}
\end{definition}
\begin{remark}
We employ the SPDE definition of the Gaussian process and the discretization in Definition \ref{def:prior} to enbale rapidly updating the prior precision matrix and its determinant for different correlation length $\gamma$. This is computationally convenient for defining MCMC samplers in Section \ref{sec:Gibbs_algorithms}.
However, the algorithms presented here can also be used for other discretisations of the Gaussian process prior, e.g., those based on the circulant embedding and the hierarchical matrices.
\end{remark}

Given the discretized prior mean and covariance, the prior distribution takes the form
\begin{gather}
\label{eq:gmrf_prior}
p_0(\param|\delta, \gamma) = (2\pi)^{-\frac{n}{2}}\,\delta^{\frac{n}{2}}\,\det\big(\prprech\big)^{\frac12} \exp\Big(-\frac{\delta}{2} \big\| \ub - \prmean \big\|^2_{\prprech}\Big).
\end{gather}
Note that the class of precision operators given by \eqref{eq:gmrf} assumes that the underlying random field is stationary up to some boundary conditions, i.e., its correlation structure is spatially invariant. One can extend the SPDE definition of the Gaussian process in \eqref{eq:gmrf} to non-stationary case, e.g., \cite{brown2018semivariogram,roininen2019hyperpriors}. We will not explore this direction in this work.

\subsection{Likelihood functions}
Given the discretized parameter $\param$, we consider the forward model in the discretized form $\predata=\forward(\param)$,
where $\predata\in\R^m$ represents the observable model outputs.
In the inverse problem, we collect measured data, denoted by $\data$, of the observables and want to estimate $\param$ from $\data$.
We use the statistical model of the measurement process and the forward model to construct the likelihood function, which takes the general form
\begin{gather}
\label{eq:like}
\like(\data|\param,\lambda) = p(\data|\predata,\lambda),%
\end{gather}
where $\lambda$ is some hyperparameter that parametrizes uncertain factors of the measurement process.
In this work, we consider two types of measurement processes.

\begin{definition}{Gaussian likelihood.}\label{def:gauss_like}
One typical assumption adopted in inverse problems is that the measurements are corrupted by zero-mean Gaussian noise. This way, we have continuous data $\data \in \R^{m}$, and the measurement process can be written as
\begin{gather*}
\data \sim \N(\predata, \lambda^{-1} \obscov),  \quad {\rm subject\;to} \quad \predata=\forward(\param),
\end{gather*}
where $ \lambda^{-1} \obscov\in\R^{m\times m}$ is a positive definite covariance matrix and $\lambda\in \R_{>0}$ is the precision parameter of the measurement process.
This leads to the likelihood function
\begin{gather}
\label{eq:like_gauss}
\like(\data|\param,\lambda) = (2\pi)^{-\frac{m}{2}}\,\lambda^{\frac{m}{2}}\,\det\big(\obscov\big)^{-\frac12}\exp\Big(-\frac{\lambda}{2}\big\Vert\forward(\param)-\data \big\Vert^2_{\obscov^{-1}}\Big).
\end{gather}
\end{definition}

\begin{definition}{Poisson likelihood.}\label{def:poisson_like}
In inverse problems such as PET imaging, the measurements are integer-valued counting data $\data \in \mathbb{N}^{m}$, and thus can be modelled by the Poisson distribution. In this setup, the expected counts of the Poisson distribution are given by the forward model $\predata = \forward(\param)$.
Each element of the observed data $\data_i$ is associated with the corresponding observable model output $\predata_i$, so the probability mass function of observing $\data_i$ is given by
\begin{gather*}
 \mathbb{P}(\data_i | \predata_i, \lambda) = \frac{\big(\lambda\,\predata_i\big)^{\data_i} \exp\big(\!-\!\lambda\,\predata_i\big) }{\data_i !},  \quad {\rm subject\;to} \quad \predata=\forward(\param),
\end{gather*}
where $\lambda \in \R_{>0}$ is a scalar variable that accounts for possible variations in the expected counts.
Assuming the measurement processes are independent, we have the likelihood function
\begin{gather}\label{eq:poisson_like}
\like(\data|\param,\lambda) = \frac{\lambda^{ {\sum_{i = 1}^{m}\data_i} }}{\prod_{i = 1}^{m} \data_i ! }\, \exp\Big( \sum_{i = 1}^{m} \big( \data_i\log\forward_i(\param)- \lambda\forward_i(\param) \big)\Big),
\end{gather}
where $\forward_i(\param)$ is the $i$-th component of the forward model outputs.
\end{definition}

\begin{assume}\label{assum:assum1}
We assume that the forward model satisfies:
\begin{enumerate}[leftmargin=0.6cm]
\item The forward model is continuously differentiable.
\item In the Gaussian likelihood case, for all $\epsilon > 0$ and $\ub$, there exists a constant $K(\epsilon) > 0$ such that
\begin{gather*}
| \forward(\ub) | \leq \exp\big( K(\epsilon) + \epsilon\|\ub\|_\mathcal{H}^2 \big) ,
\end{gather*}
where $\|\ub\|_\mathcal{H}$ is some appropriate discretized function norm.
\item For the Poisson likelihood, $\forward(\ub)$ is non-negative and bounded, i.e., $\forward(\ub) \in \R^m_{\geq 0}$ and $| \forward(\ub) | < \infty$.
\end{enumerate}
\end{assume}

The continuous differentiability assumption (Condition 1 of the above assumption) implies the Lipschitz continuity of the forward model.
Condition 2 of the above assumption is to ensure that the forward model is sufficiently bounded in the Gaussian likelihood case (see \cite{IP:Stuart_2010} for details).
Condition 3 of the above assumption is necessary for the Poisson likelihood case, since the model outputs give the expected counts of the Poisson distribution.
These are sufficient conditions that can be used in the framework of \cite{IP:Stuart_2010} to define well-posed Bayesian inverse problems in the function space setting.
In this work, although we focus on the computation of the hierarchical Bayesian inverse problem, we keep our problem setup consistent to that of \cite{IP:Stuart_2010}.

\subsection{Joint posterior and its marginal and conditional distributions}\label{sec:post_def}
In the case that hyperparameters $\lambda$, $\delta$, and $\gamma$ are unknown, in keeping with the Bayesian paradigm, we assume hyper-priors $p_0(\lambda)$, $p_0(\delta)$, and $p_0(\gamma)$ on those hyperparameters. Thus, the {\it joint posterior density} over all of the unknown parameters is given, by Bayes' law, as
\begin{gather}
\label{eq:joint_post}
p(\param,\lambda, \delta, \gamma|\data) = \frac{1}{p(\data)}\like(\data|\param,\lambda)\,p_0(\param|\delta, \gamma)\,p_0(\lambda)\,p_0(\delta)\,p_0(\gamma),
\end{gather}
where $ p(\data) $ is the normalizing constant. %
The following densities associated with the joint posterior density will be used throughout this paper:
In some situations, we use the {\it marginal posterior density}
\begin{gather}
p(\lambda, \delta, \gamma|\data) = \frac{1}{p(\data)}\like(\data|\lambda, \delta, \gamma)\,p_0(\lambda)\,p_0(\delta)\,p_0(\gamma),
\label{eq:marginal_post}
\end{gather}
to draw hyperparameter samples.
To define the marginal posterior, we need the {\it marginal likelihood}:
\begin{gather}
\like(\data|\lambda, \delta, \gamma) = \int \like(\data|\param,\lambda)\,p_0(\param|\delta, \gamma)\, d\param.
\label{eq:marginal_like}
\end{gather}
The marginal likelihood is also the normalising constant of the {\it conditional posterior density}
\begin{gather}
p(\param|\data,\lambda, \delta, \gamma) = \frac{1}{\like(\data|\lambda, \delta, \gamma)}\like(\data|\param,\lambda)\,p_0(\param|\delta, \gamma).
\label{eq:cond_post}
\end{gather}
The conditional posterior is used to draw parameter samples. Since the marginal likelihood is often unknown, we can only evaluate the {\it unnormalized conditional posterior density} 
\begin{gather}
f(\param|\data,\lambda, \delta, \gamma) = \like(\data|\param,\lambda)\,p_0(\param|\delta, \gamma).
\label{eq:un_cond_post}
\end{gather}

%% file: sec_RTO_a.tex
\section{Randomize-then-optimize}
\label{sec:RTO}
In this work, the foundation for exploring the joint posterior distribution \eqref{eq:joint_post} relies on the capability of sampling the conditional posterior \eqref{eq:cond_post} and estimating the marginal likelihood \eqref{eq:marginal_like}.
For prescribed hyperparameters $\lambda, \delta, \gamma$, we seek to construct a map
\begin{gather}
\label{eq:rtomap1}
\Tb(\ub; \lambda, \delta, \gamma) = \zetab, %
\end{gather}
where $\Tb : \R^n \rightarrow \R^n$ and $\zetab \in \R^n$, to describe a deterministic coupling between a reference random variable $\zetab \sim p_{\rm ref}(\cdot |  \lambda, \delta, \gamma)$ and some target random variable $\ub$.
If the probability density of the random variable $\ub$ is a good approximation to the conditional posterior, then we can use the coupling \eqref{eq:rtomap1} to efficiently sample the conditional posterior and/or to compute the marginal likelihood through importance sampling.
Then, we can apply either the Metropolis-within-Gibbs method (see \cite{agapiou2014analysis} and references therein) or the pseudo-marginal method \cite{andrieu2009pseudo} to explore the joint posterior.

In this section, we will first present an overview of the scalable randomize-then-optimize method \cite{wang2019scalable} that can be applied to build the coupling in \eqref{eq:rtomap1} for the Gaussian likelihood (see Definition \ref{def:gauss_like}) in the hierarchical Bayesian setting.
Then, we will show that RTO satisfies the Central Limit Theorem for estimating the marginal likelihood under certain technical assumptions, present a trust-region modification that enables that RTO can be correctly implemented under relaxed technical assumptions, and generalize RTO to problems with Poisson likelihood.

\newcommand{\ubref}{\ub_{\lambda, \delta, \gamma}^\ast}
\renewcommand{\ubref}{\ub_\ast}

\subsection{RTO for Gaussian likelihood}\label{sec:gauss_rto}
Suppose hyperparameters $\lambda, \delta, \gamma$ are prescribed. We have a Guassian prior $\N(\prmean, \delta^{-1}\prprech^{-1} )$ and a Gaussian measurement process $\data \sim \N(\predata, \lambda^{-1} \obscov)$ subject to $\predata=\forward(\param)$.
RTO uses elements of the conditional posterior, including the linearized forward model $\Jb(\ub) := \nabla_\ub \Fb(\ub) \in \R^{m \times n}$, the prior precision matrix $\delta \prprech$, and the covariance of the measurement noise $\lambda^{-1} \obscov$, to construct the coupling equation.
Suppose we have a reference parameter $\ubref$ for fixed $\lambda, \delta, \gamma$, e.g., the maximum {\it a posteriori} (MAP) point or the posterior mean, and the matrix $\Jb(\ubref)$ is rank-$r$ where $r \leq \min(m,n)$. We compute the compact generalized SVD:
\begin{gather}
\label{eq:gsvd}
\frac{\lambda}{\delta} \, \obscov^{-1} \, \Jb(\ub_\ast) \, \prprech^{-1} = \Yb \Sb \Xb^\top,
\end{gather}
where $\Sb \in \R^{r \times r}$ is diagonal, $\Yb \in \R^{m \times r}$ is $(\lambda^{-1} \obscov)$-orthogonal, and $\Xb \in \R^{n \times r}$ is $(\delta \prprech)$-orthogonal. That is, $ \Yb^\top (\lambda^{-1} \obscov) \Yb = \Ib_r$ and $\Xb^\top (\delta \prprech) \Xb = \Ib_r$.
Then, introducing a $(\delta\prprech)$-orthogonal projector $\Pib = \Xb \Xb^\top (\delta \prprech)$, the scalable RTO (see Section 4.2 of \cite{wang2019scalable}) constructs a nonlinear function
\begin{gather}
\label{eq:oblique}
\Tb(\ub; \lambda, \delta, \gamma)  \coloneqq \Xb \big[ (\Sb^2 + \Ib)^{-\frac12} \big( \Xb^\top (\delta \prprech)(\ub-\prmean) + \Sb \Yb^\top (\Fb(\ub)-\data) \big) \big] + (\Ib-\Pib) (\ub-\prmean).%
\end{gather}
Given a zero-mean Gaussian random variable $\zetab$ with the precision matrix $\delta\prprech$, the scalable RTO then defines the coupling equation
\begin{gather}
\label{eq:rtomap}
\Tb(\ub; \lambda, \delta, \gamma) = \zetab, \quad \text{where} \quad \zetab \sim \N(\zerob, (\delta\prprech)^{-1}).
\end{gather}

We can decompose the parameter $\ub$ into three parts: the prior mean, one in the column space of $\Xb$, and another in the $(\delta\prprech)$-orthogonal complement of $\Xb$.
Defining
\begin{gather*}
\ub_r = \Xb^\top (\delta\prprech) (\ub-\prmean), \quad \ub_\perp = (\Ib - \Pib) (\ub-\prmean), \quad \text{and}, \quad \ub = \Xb \ub_r + \ub_\perp + \prmean,
\end{gather*}
and a nonlinear function
\begin{gather*}
\Theta(\ub_r; \ub_\perp) = (\Sb^2 + \Ib)^{-\frac12}\big( \ub_r + \Sb \Yb^\top \big(\Fb\big(\Xb \ub_r + \ub_\perp + \prmean\big)-\data\big) \big) \in \R^r
\end{gather*}
the nonlinear equation \eqref{eq:rtomap} can be written as a coupled system of equations:
\begin{gather}
\label{eq:split}
\left\{
\begin{aligned}
\ub_\perp  & = (\Ib-\Pib) \, \zetab \\
\Xb  \Theta(\ub_r; \ub_\perp) & = \Pib \, \zetab
\end{aligned}
\right.
\end{gather}
Thus, we can solve the nonlinear system of equations \eqref{eq:rtomap} by first computing $\ub_\perp = (\Ib - \Pib) \, \zetab$ and then solving $(\Sb^2 + \Ib)^{-\frac12}\Theta(\ub_r; \ub_\perp) = \Xb^\top (\delta \prprech)\,\zetab $ via the $r$-dimensional optimization problem
\begin{gather}\label{eq:svd2}
 \ub_r = \argmin_{u_r'} \big\| \Theta(\ub_r'; \ub_\perp) - \Xb^\top (\delta \prprech)\,\zetab\big\|^2.
 \end{gather}

\begin{theorem}\label{prop:rto}
In addition to Assumption \ref{assum:assum1}, we assume that for all $\ub_r \in \R^r$ and $\ub_\perp \in {\rm kernel}(\Xb)$, the mapping $\ub_r \mapsto \Theta(\ub_r; \ub_\perp)$ is invertible.
Then, the normalized probability density of $\ub$ generated by the solving coupling equation \eqref{eq:rtomap} is given by the pullback density of $\N(\zerob, \prprech^{-1})$ under the mapping $\Tb(\ub; \lambda, \delta, \gamma) = \zetab$:
\begin{gather}
p_{\rm RTO}(\ub| \lambda, \delta, \gamma) = (2\pi)^{-\frac{n}{2}}\delta^{\frac{n}{2}}\,\det\big(\prprech\big)^{\frac12}
\det\big(\nabla_{\ub}\Tb(\ub; \lambda, \delta, \gamma)\big)  \exp\Big(-\frac{\delta}{2} \big\| \Tb(\ub; \lambda, \delta, \gamma) \big\|_{\prprech}^2 \Big).
\label{eq:rto_density}
\end{gather}
Then, the ratio between the unnormalized posterior and the RTO density is given by
\begin{align}
\!\!w(\ub; \lambda, \delta, \gamma) \!:= & \frac{f(\ub| \data,\lambda, \delta, \gamma)}{p_{\rm RTO}(\ub| \lambda, \delta, \gamma)} \nonumber \\
\! = & \frac{\!\!\lambda^{\frac{m}{2}}\!\det\big(\obscov\big)^{\!\!-\frac12}\!\!\det\big(\Sb^2 \!+\! \Ib_r\big)^{\frac12}}{(2\pi)^{\!\frac{m}{2}}\!\det\big( \Ib_r \!+\! \Sb \Yb^\top \nabla_{\ub} \Fb\big( \ub \big) \Xb \big)\!\!}
\exp\!\Big(\!\!\!-\!\frac{\lambda}{2}\!\big\Vert\forward(\param)-\data \big\Vert^2_{\obscov^{-1}} \!-\! \frac{1}{2}\! \left\| \ub_r \right\|^2 \!+\! \frac{1}{2}\! \big\| \Theta(\ub_r; \ub_\perp) \big\|^2 \!\Big).\!\!\!\!
\label{eq:rto_weight}
\end{align}
In addition, the ratio $w(\ub; \lambda, \delta, \gamma)$ is positive almost surely w.r.t. the prior $\pi(\ub | \delta, \gamma)$.
\end{theorem}
\begin{proof}
The continuous differentiability assumption of the forward model (Assumption \ref{assum:assum1}) implies that the mapping $\Theta$ is also continuously differentiable. Together with invertibility assumption above, the mapping $\Theta$ is diffeomorphic.
The rest of the above results are equivalent to Proposition 3 and Theorem 8 of \cite{wang2019scalable}. For completeness, we provide the derivation of the RTO coupling equation \eqref{eq:rtomap} and the RTO density in \eqref{eq:rto_density} in Appendix \ref{sec:app_a1} to formally establish this equivalence.
\end{proof}

The RTO mapping in \eqref{eq:rtomap} can be used either as an independence proposal in the Metropolis-Hastings (MH) algorithm or as a biasing density in the importance sampling for exploring the conditional posterior.
In both cases, the weighting function $w(\ub; \lambda, \delta, \gamma)$ in \eqref{eq:rto_weight} can be used to define either the acceptance probability in MH or the importance ratio in importance sampling.
Algorithm \ref{alg:rto_mh} summarizes the procedure of simulating a Markov chain with the conditional posterior as the invariant density for $N$ number of steps. In this algorithm, generating RTO samples is the most computationally demanding part; fortunately all RTO samples can be generated in parallel.
Algorithm \ref{alg:rto_is} summarizes the procedure of computing the marginal likelihood using importance sampling and RTO.
We will exploit these features in later sections for exploring the joint posterior.

\begin{algorithm}[htbp]
\caption{RTO-Metropolis-Hastings for sampling from $p(\ub|\data,\lambda,\delta,\gamma)$} \label{alg:rto_mh}{}
\begin{algorithmic}[1]
\State For fixed $\lambda, \delta, \gamma$, find the reference parameter $\ub_\ast$, e.g., the MAP point.
\State Compute the generalized SVD \eqref{eq:gsvd} of the linearized forward model $\Jb(\ub_\ast)$ to define the coupling equation in \eqref{eq:oblique}.
\For {$i = 0, \ldots, N$} in parallel
  \State Draw random variables $\etab_{i} \sim \N(\zerob, \prprech^{-1})$.
  \State Solve for a corresponding RTO sample $\ub_{i} = \Xb \ub_{r,i} + \ub_{\perp, i} + \prmean$ using \eqref{eq:split}.
  \State Compute the weights $w(\ub_{i}; \lambda, \delta, \gamma)$ using \eqref{eq:rto_weight}.
\EndFor
\State Initialise the Markov chain $\Ub_{0} = \ub_{0}$.
\For {$i = 1, \ldots, N$} in series
\State With probability $\alpha(\ub_i,\ub_{i-1})=\min\{1,w(\ub_{i-1}; \lambda, \delta, \gamma) / w(\ub_{i}; \lambda, \delta, \gamma)\}$, {\bf accept} $\ub_i$ by seting $\Ub_{i}$ = $\ub_{i}$, otherwise, {\bf reject} by setting $\Ub_{i}$ = $\Ub_{i-1}$.
\EndFor
\end{algorithmic}
\end{algorithm}

\begin{algorithm}[htbp]
\caption{RTO importance sampling for computing the marginal likelihood} \label{alg:rto_is}{}
\begin{algorithmic}[1]
\State For fixed $\lambda, \delta, \gamma$, find the reference parameter $\ub_\ast$, e.g., the MAP point.
\State Compute the generalized SVD \eqref{eq:gsvd} of the linearized forward model $\Jb(\ub_\ast)$ to define the coupling equation in \eqref{eq:oblique}.
\State Compute RTO samples and weights, $\{\ub_{i},w(\ub_{i}; \lambda, \delta, \gamma)\}_{i=1}^{N}$, as in Algorithm \ref{alg:rto_mh}.
\State Approximate the marginal likelihood by importance sampling:\vspace{-1em}
\begin{gather}\label{eq:marg_like_PM}
\!\!\like(\data|\lambda, \delta, \gamma) \! = \! \mathbb{E}_{p_{\rm RTO}} \big[ w(\ub; \lambda, \delta, \gamma) \big] \! \approx\! \like_N(\data|\lambda, \delta, \gamma) \! \coloneqq\!  \frac{1}{N} \sum_{i = 1}^{N}w(\ub_{i}; \lambda, \delta, \gamma).\!\!
\end{gather}
\vspace{-1em}
\end{algorithmic}
\end{algorithm}

\begin{remark}\label{rem:general_eig}
The generalized SVD in \eqref{eq:gsvd} is equivalent to the generalized eigenvalue problems:
\begin{gather}\begin{aligned}
\big( \Jb(\ub_\ast)^\top \,(\lambda \obscov^{-1}) \, \Jb(\ub_\ast) \big)  \Xb & = \big(\delta \prprech\big) \, \Xb \, \Sb^2 , \\
\big( \Jb(\ub_\ast) \,(\delta^{-1} \prprech^{-1}) \,  \Jb(\ub_\ast)^\top \big)  \Yb & = \big(\lambda^{-1} \obscov\big) \, \Yb \, \Sb^2 \label{eq:geig}.
\end{aligned}\end{gather}
For problems where it is not feasible to explicitly construct the linearized forward model, matrix-free solvers such as Lanczos or randomized SVD (see \cite{Lin:GoVanLo_2012,Lin:HMT_2011} and references therein) can be used to solve the generalized eigenvalue problems \eqref{eq:geig} to obtain $\Xb$, $\Yb$, and $\Sb$.
\end{remark}

\subsection{Using RTO in importance sampling}
In this work, one important application of RTO is to compute the marginal likelihood as outlined in Algorithm \ref{alg:rto_is}.
The result of Theorem \ref{prop:rto} ensures that the importance sampling estimator $\like_N(\data|\lambda, \delta, \gamma) $ in \eqref{eq:marg_like_PM} satisfies the Strong Law of Large Numbers (see Chapter 9 \cite{mcbook}). Thus Algorithm \ref{alg:rto_is} provides an almost surely converging and unbiased estimate of the marginal likelihood.
The following proposition establishes that $\like_N(\data|\lambda, \delta, \gamma) $ also satisfies the Central Limit Theorem.

\begin{proposition}\label{prop:rto_is}
Under the same assumption of Proposition \ref{prop:rto}, the second moment of the ratio between the unnormalized posterior and the RTO density, $w(\ub; \lambda, \delta, \gamma)$, is finite. That is,
\begin{gather*}
\mathbb{E}_{p_{\rm RTO}}\big[ w(\ub; \lambda, \delta, \gamma)^2\big] < \infty.
\end{gather*}
\end{proposition}

\begin{proof}
See Appendix \ref{sec:app_a2}.
\end{proof}

\subsection{Trust-region modification to RTO}
\label{sec:trust_rto}
For many practical problems, the global diffeomorphism assumption of the mapping $\ub_r \mapsto \Theta(\ub_r; \ub_\perp)$ in Theorem \ref{prop:rto} is often hard to satisfy.
Here we propose a trust-region modification to the forward model to provide heuristics that may guarantee this assumption in practice.
We first split the mapping $\Theta$ into the first order Taylor series expansion around $\ub_\ast$ and the remainder:
\begin{gather}\label{eq:map_split}
\Theta(\ub_r; \ub_\perp) = \Theta_{\rm L}(\ub_r; \ub_\perp) + \Theta_{\rm R}(\ub_r; \ub_\perp),
\end{gather}
where the first order Taylor series expansion is given by
\begin{gather}\begin{aligned}
\Theta_{\rm L}(\ub_r; \ub_\perp) & = \ub_r + \Sb \Yb^\top \Big(  \forward(\ub_\ast) + \Jb(\ub_\ast)(\Xb \ub_r + \ub_\perp + \prmean - \ub_\ast)  - \data \Big)
\label{eq:map_linear},
\end{aligned}\end{gather}
and the remainder term is given by
\begin{gather}\label{eq:map_nonlinear}
\Theta_{\rm R}(\ub_r; \ub_\perp) = \Sb \Yb^\top \Big( \forward(\Xb \ub_r + \ub_\perp + \prmean) - \Jb(\ub_\ast)(\Xb \ub_r + \ub_\perp + \prmean - \ub_\ast) - \forward(\ub_\ast) \Big).
\end{gather}
The Jacobian matrices of $\Theta_{\rm L}(\ub_r; \ub_\perp)$ and the remainder can be respectively expressed as
\begin{gather*}\begin{aligned}
\nabla_{\ub_r} \Theta_{\rm L}(\ub_r; \ub_\perp) & = \Ib_r + \Sb \Yb^\top \Jb(\ub_\ast) \Xb = \Ib_r + \Sb^2, \\ %
\nabla_{\ub_r} \Theta_{\rm R}(\ub_r; \ub_\perp) & = \Sb \Yb^\top \big( \Jb(\Xb \ub_r + \ub_\perp + \prmean) - \Jb(\ub_\ast)\big) \Xb,
\end{aligned}\end{gather*}
where the first identity follows from the generalized SVD in \eqref{eq:gsvd}.
Our starting point is that the spectral radius of $\nabla_{\ub_r} \Theta_{\rm R}(\ub_r; \ub_\perp)$ can be locally bounded w.r.t. $\nabla_{\ub_r} \Theta_{\rm L}(\ub_r; \ub_\perp)$, as stated below.
\begin{assume}\label{assum:assum2}
For a given reference point $\mb_r = \Xb^\top (\delta \prprech) (\ub_\ast - \prmean)$, we assume that there exists a constant $\varepsilon \in \R_{>0}$ and a set $\,\mathbb{S}(\varepsilon) = \{ \ub_r \in \R^r : \| \ub_r - \mb_r \| < \varepsilon\}$ such that the largest singular value of the matrix $\nabla_{\ub_r}\!\! \Theta_{\rm R}(\ub_r; \ub_\perp) (\Ib_r + \Sb^2)^{-1}$ is bounded below $1$ for $\forall \ub_r \in \mathbb{S}(\varepsilon)$ and $\forall \ub_\perp \in {\rm kernel}(\Xb) $, that is
\begin{gather*}
\sup_{\ub_r \in \mathbb{S}(\varepsilon), \ub_\perp \in {\rm kernel}(\Xb)} \sigma_{\rm max}\Big( (\Ib_r + \Sb^2)^{-1}\nabla_{\ub_r}\!\!\Theta_{\rm R}(\ub_r; \ub_\perp) \Big) < 1.
\end{gather*}
\end{assume}
Following Assumption \ref{assum:assum2}, the mapping $\ub_r \mapsto \Theta(\ub_r; \ub_\perp)$ is locally diffeomorphic for all  $\ub_r \in \mathbb{S}(\varepsilon)$.
We want to extend this local diffeomorphism to $\R^r$ by applying a nonlinear transformation to the remainder term in \eqref{eq:map_nonlinear}.
Towards this goal, we introduce a trust region function $\psi: \R_{\geq 0} \mapsto \R$:
\begin{gather}
\psi( r ; \tilde\varepsilon, \tau) =
\begin{cases}
r & \text{if} \quad r < \tilde\varepsilon(1-\tau) \vspace{-0.3em} \\
\tilde\varepsilon - \frac{\tau\tilde\varepsilon}{4} + \frac{r - \tilde\varepsilon}{2} - \frac{(r - \tilde\varepsilon)^2}{4\,\tau\tilde\varepsilon} & \text{if} \quad r \in [\tilde\varepsilon(1-\tau), \tilde\varepsilon(1+\tau) ) \vspace{-0.2em} \\
\tilde\varepsilon & \text{if} \quad r \geq \tilde\varepsilon(1+\tau)
\end{cases}
\end{gather}
where $\tilde\varepsilon > 0$ and $\tau \in (0, 1)$. The function $\psi( r ; \tilde\varepsilon, \tau)$ satisfies three conditions: ({\romannumeral 1}) it is bounded, i.e., $0 \leq \psi(r) \leq \tilde\varepsilon$ for $\forall r \in \R_{\geq 0}$, ({\romannumeral 2}) it is first order continuous, i.e., $\psi \in \mathbb{C}_1$; and ({\romannumeral 3}) its derivative is non-negative, bounded, and vanishing at the tails, that is, $\psi^\prime(r) \in [0, 1]$ for all $r$, and $\psi^\prime (r) = 0$ for $r \geq \tilde\varepsilon(1+\tau)$.
Then, we construct a smooth nonlinear transformation $\Psib : \R^r \mapsto \R^r$:
\begin{gather}
\Psib(\ub_r; \tilde\varepsilon, \tau) = \mb_r + \frac{\psi(\|\ub_r - \mb_r\|; \tilde\varepsilon, \tau)}{\|\ub_r - \mb_r\|} \big( \ub_r -\mb_r \big),
\end{gather}
which transform $\ub_r \in \R^r$ to the set $\mathbb{S}\big(\tilde\varepsilon\big)$. For all $\ub_r \in \mathbb{S}\big(\tilde\varepsilon(1-\tau)\big)$, we simply have $\Psib(\ub_r; \tilde\varepsilon, \tau) = \ub_r-\mb_r$. For any $\ub_r$ in the complement of $\mathbb{S}\big(\tilde\varepsilon(1-\tau)\big)$, the transformation smoothly warps $\ub_r$ into the set $\mathbb{S}(\tilde\varepsilon)$ along the normal direction defined by $\|\ub_r - \mb_r\|^{-1}( \ub_r -\mb_r)$.

\begin{proposition}\label{prop:trust_region}
Suppose that we have an original mapping satisfies Assumption \ref{assum:assum2}. Given $\tilde\varepsilon\leq \varepsilon$ and $0< \tau \ll 1$, we construct an alternative mapping
\begin{gather}
\widetilde\Theta(\ub_r; \ub_\perp) = \Theta_{\rm L}(\ub_r; \ub_\perp) + \Theta_{\rm R}\Big( \Psib(\ub_r; \tilde\varepsilon, \tau); \ub_\perp)\Big),
\end{gather}
Then, the modified mapping $\ub_r \mapsto \widetilde\Theta(\ub_r; \ub_\perp)$ is diffeomorphic for all $\ub_r \in \R^r$ and $\ub_\perp \in {\rm kernel}(\Xb)$.
\end{proposition}
\begin{proof}
See Appendix \ref{sec:app_a3}.
\end{proof}

\subsection{RTO for Poisson likelihood}
The RTO formulation presented in Proposition \ref{prop:rto} is limited to problems with Gaussian prior and Gaussian observation noise. By transforming non-Gaussian prior densities into Gaussian densities, e.g., \cite{chen2018robust,wang2017bayesian}, this Gaussian prior limitation may be relaxed. 
We employ the importance sampling principle here to present a RTO formulation that can be applied to the Poisson likelihood. %

For fixed hyperparameters $\lambda, \delta, \gamma$, we use the following Gaussian likelihood to approximate the Poisson likelihood, and hence to define the RTO importance density.
We first express the logarithm of the Poisson likelihood as a function of the logarithm of the observable model outputs:
\begin{gather}\begin{aligned}\label{Poisson_log_like}
\log \like(\data|\param,\lambda) = -  \sum_{i = 1}^{m} \log \data_i! - \lambda\,\sum_{i = 1}^{m} \exp(\xib_i) + \sum_{i = 1}^{m} \data_i \big( \log\lambda + \xib_i\big) \;\; {\rm subject\;to} \;\; \xib=\log\forward(\param).
\end{aligned}\end{gather}
Given a reference parameter $\ub_\ast$, we expand $\log \like(\data|\ub,\lambda)$ in a second-order Taylor series about $\xib_\ast=\log\Fb(\ub_\ast)$ and move the higher order terms into the error to obtain
\begin{gather}\begin{aligned}
\log \like(\data_\ast|\ub,\lambda) =& \log \like(\data_\ast|\ub_\ast,\lambda) + (\xib-\xib_\ast)^\top \nabla_{\xi} \log \like(\data_\ast|\ub_\ast,\lambda) \\
& +\frac12(\xib-\xib_\ast)^\top \nabla_{\xib}^2 \log \like(\data_\ast|\ub_\ast,\lambda) (\xib-\xib_\ast) +\mathcal{O}(\Vert\xib-\xib_\ast\Vert^3) \\
=& \log \like(\data_\ast|\ub_\ast,\lambda) - \frac{\lambda}{2} \big\Vert\xib-\data_\ast \big\Vert^2_{\obscov_\ast^{-1}} +\mathcal{O}\big(\big\Vert\xib-\xib_\ast\big\Vert^3\big), \label{eq:like_taylor}
\end{aligned}\end{gather}
where $\data_\ast = \log(\data/\lambda)$ and $\obscov_\ast^{-1}={\rm diag}(\forward(\ub_\ast))$. Dropping the error term in the Taylor series and applying the identity $\xib= \log\Fb(\param)$, we obtain the Gaussian surrogate likelihood
\begin{gather}
\label{eq:s_like}
\log  \like(\data_\ast|\param,\lambda)\approx {\rm const} - \frac{\lambda}{2} \big\Vert \log \Fb(\param)-\data_\ast \big\Vert^2_{\obscov_\ast^{-1}}
\end{gather}
From (\ref{eq:s_like}), we obtain the biasing conditional posterior
\begin{gather}\begin{aligned}
\pi_\ast(\param|\data_\ast,\lambda, \delta, \gamma) & \propto f_\ast(\param|\data_\ast,\lambda, \delta, \gamma) \\
& =
(2\pi)^{-\frac{m}{2}}\,\lambda^{\frac{m}{2}}\,\det\big(\obscov_\ast\big)^{-\frac12}\exp\Big(- \frac{\lambda}{2} \big\Vert \log \Fb(\param)-\data_\ast \big\Vert^2_{\obscov_\ast^{-1}}\Big)\,p_0(\param|\delta, \gamma). \label{eq:s_cond_post}
\end{aligned}\end{gather}

\begin{corollary}
Defining the RTO coupling equation $\Theta(\ub_r; \ub_\perp)$ for sampling the biasing conditional posterior $\pi_\ast(\param|\data_\ast,\lambda, \delta, \gamma)$, the ratio between the unnormalized posterior (with the Poisson likelihood) and the associated RTO density is given by
\begin{gather}\begin{aligned}\label{eq:rto_weight_poisson}
 w(\ub; \lambda, \delta, \gamma) = &  \frac{\det\big(\Sb^2 \!+\! \Ib\big)^{\frac12}\, \lambda^{ {\sum_{i = 1}^{m}\data_i} }}{\det\big( \Ib_r \!+\! \Sb \Yb^\top \nabla_{\ub} \Fb\big( \ub \big) \Xb \big)\prod_{i = 1}^{m} \data_i !}\\
 & \exp\Big( \sum_{i = 1}^{m} \big( \data_i\log\forward_i(\param)- \lambda\forward_i(\param) \big)- \frac{1}{2} \big\| \ub_r \big\|^2 + \frac{1}{2} \big\| \Theta(\ub_r; \ub_\perp) \big\|^2 \Big).
\end{aligned}\end{gather}
The ratio $w(\ub; \lambda, \delta, \gamma)$ is positive almost surely w.r.t. the prior $\pi(\ub | \delta, \gamma)$ and has finite second moment, i.e., $\mathbb{E}_{p_{\rm RTO}}\big[ w(\ub; \lambda, \delta, \gamma)^2 \big] < \infty$.
\end{corollary}
\begin{proof}
The ratio between the unnormalized conditional posterior and the RTO density is
\begin{gather*}
w(\param;\lambda, \delta, \gamma) = \frac{f(\param|\data,\lambda, \delta, \gamma)}{\pi_\ast(\param|\data_\ast,\lambda, \delta, \gamma)} \frac{\pi_\ast(\param|\data_\ast,\lambda, \delta, \gamma)}{p_{\rm RTO}(\param|\lambda, \delta, \gamma)} = \frac{f(\param|\data,\lambda, \delta, \gamma)}{f_\ast(\param|\data_\ast,\lambda, \delta, \gamma)} \frac{f_\ast(\param|\data_\ast,\lambda, \delta, \gamma)}{p_{\rm RTO}(\param|\lambda, \delta, \gamma)}.
\end{gather*}
Writing $w_\ast(\param|\lambda, \delta, \gamma) = f_\ast(\param|\data_\ast,\lambda, \delta, \gamma)\big/p_{\rm RTO}(\param|\lambda, \delta, \gamma)$ as in \eqref{eq:rto_weight}, the ratio $w$ can be written as
\begin{gather*}
w(\param;\lambda, \delta, \gamma) \propto \exp\Big( \frac{\lambda}{2} \Big\Vert\log\forward(\param)-\data_\ast \Big\Vert^2_{\obscov_\ast^{-1}} + \sum_{i = 1}^{m}  \big( \data_i\log\forward_i(\param)- \lambda\forward_i(\param) \big) \Big) w_\ast(\param;\lambda, \delta, \gamma) .
\end{gather*}
Given Condition 3 of Assumption \ref{assum:assum1}, there exist constants $c_1, c_2 > 0$ such that
\begin{gather*}
c_1 < \exp\Big( \frac{\lambda}{2} \Big\Vert\log\forward(\param)-\data_\ast \Big\Vert^2_{\obscov_\ast^{-1}} + \sum_{i = 1}^{m}  \big( \data_i\log\forward_i(\param)- \lambda\forward_i(\param) \big) \Big) < c_2.
\end{gather*}
Then, the results directly follows from Theorem \ref{prop:rto} and Proposition \ref{prop:rto_is}.
\end{proof}

\begin{remark}
We can use Taylor series expansions w.r.t. different variables to construct the biasing conditional posterior. For example, one can expand w.r.t. $\predata = \forward(\ub)$ instead of $\xib = \log\forward(\ub)$ used here. We choose the Taylor series expansion in \eqref{eq:like_taylor} to reduce the nonlinearity of the forward model used in PET imaging. See Section \ref{sec:PET} for details.
\end{remark} 

%% file: sec_Gibbs_a.tex
\section{The full hierarchical model and RTO-within-Gibbs}
\label{sec:Gibbs_algorithms}

In this section, we extend the hierarchical Gibbs sampler of \cite{bardsley2012mcmc} to sample from the joint posterior (\ref{eq:joint_post}) with nonlinear forward models. To accomplish this, we employ a Metropolis-within-Gibbs strategy with RTO as the proposal distribution and present a new computationally fast way to update the hyperparameter $\gamma$ that controls the correlation length of the prior. 

\subsection{Hyper-prior}
We first define the hyper-priors $p_0(\lambda)$, $p_0(\delta)$, and $p_0(\gamma)$ to fully specify the joint posterior distribution.
\begin{definition}\label{def:hyper}
Following the setup of \cite{bardsley2012mcmc}, we use Gamma distributions as hyper-priors for $\lambda$ and $\delta$, that is, $p_0(\lambda) = \Gamma(\alpha_{\lambda}, \beta_{\lambda})$ and $p_0(\delta) = \Gamma(\alpha_{\delta}, \beta_{\delta})$, which have the density functions 
\begin{gather*}
p_0(\lambda)\propto \lambda^{\alpha_{\lambda}-1}\exp(-\beta_{\lambda}\lambda),\quad{\rm and}, \quad
p_0(\delta)\propto \delta^{\alpha_{\delta}-1}\exp(-\beta_{\delta}\delta),
\end{gather*}
respectively. For $p_0(\gamma)$, we assume a Beta hyper-prior distribution scaled to the domain $[\gamma_{\rm L}, \gamma_{\rm R}]$:
\begin{gather}
p_0(\gamma) \propto \mathds{1}_{[\gamma_{\rm L}, \gamma_{\rm R}]}(\gamma)\,(\gamma - \gamma_{\rm L})^{\alpha_{\gamma}} (\gamma_{\rm R} -\gamma)^{\beta_{\gamma}},
\end{gather}
where $\mathds{1}_{[\gamma_{\rm L}, \gamma_{\rm R}]}(\gamma)$ is the indicator function.
\end{definition}

\begin{definition}
For the Gaussian likelihood function, we have the joint posterior density
\begin{gather}
\begin{aligned}
\label{eq:full_gauss_post}
p(\param,\lambda, \delta, \gamma|\data) & \propto  \mathds{1}_{[\gamma_{\rm L}, \gamma_{\rm R}]}(\gamma) \, (\gamma - \gamma_{\rm L})^{\alpha_{\gamma}} (\gamma_{\rm R} -\gamma)^{\beta_{\gamma}} \; \lambda^{\alpha_\lambda-1+\frac{m}2} \; \delta^{\alpha_\delta-1+\frac{n}2}\; \det \big( \prprech \big)^{\frac12}  \\
& \quad\;\; \exp\Big(-\frac\lambda2\Big\Vert\forward(\param)-\data\Big\Vert^2_{\obscov^{-1}}
-\frac{\delta}{2} \big\Vert\ub-\prmean\big\Vert^2_{\prprech} -\beta_{\lambda}\lambda -\beta_{\delta}\delta\Big).
\end{aligned}
\end{gather}
For the Poisson likelihood function, we have the joint posterior density
\begin{gather}
\begin{aligned}
\label{eq:full_poisson_post}
p(\param,\lambda, \delta, \gamma|\data) & \propto  \mathds{1}_{[\gamma_{\rm L}, \gamma_{\rm R}]}(\gamma) \, (\gamma - \gamma_{\rm L})^{\alpha_{\gamma}} (\gamma_{\rm R} -\gamma)^{\beta_{\gamma}} \; \lambda^{\alpha_\lambda-1+{\sum_{i = 1}^{m}\data_i}} \; \delta^{\alpha_\delta-1+\frac{n}2} \det \big( \prprech \big)^{\frac12} \\
& \quad\;\;  \exp\Big( - \lambda\,\sum_{i = 1}^{m} \forward_i(\param) + \sum_{i = 1}^{m} \data_i \log \forward_i(\param)
-\frac{\delta}{2} \big\Vert\ub-\prmean\big\Vert^2_{\prprech} -\beta_{\lambda}\lambda -\beta_{\delta}\delta\Big).
\end{aligned}
\end{gather}
\end{definition}

\subsection{RTO-within-Gibbs}
A straightforward, at least in theory, MCMC method for sampling from the joint posterior distribution $p(\param,\lambda, \delta, \gamma|\data)$ is to use the Gibbs sampler that cyclically samples from the conditional densities $p(\ub|\yb,\lambda,\delta,\gamma)$, $p(\lambda, \delta |\yb, \ub, \gamma)$, and $p(\gamma|\yb, \ub, \alpha)$:

\paragraph{1. Updating $\ub$ given $(\lambda, \delta, \gamma)$} Since in the nonlinear case it is not possible to sample directly from $p(\ub|\yb,\lambda,\delta,\gamma)$, we use RTO as a MH proposal (see Algorithm \ref{alg:rto_mh}) to update $\ub$ for fixed $(\lambda, \delta, \gamma)$.

\paragraph{2. Updating $(\lambda, \delta)$ given $(\ub, \gamma)$} For fixed $\ub$ and $\gamma$, the conditional density $p(\lambda, \delta |\yb, \ub, \gamma)$ can be written as the product of two Gamma distributions, which take the form
\begin{gather}\label{eq:lambda_delta_conditional}
p(\lambda, \delta |\yb, \ub, \gamma) = p(\lambda|\yb, \ub, \gamma)\,p(\delta |\yb, \ub, \gamma).
\end{gather}
This way, for updating $\lambda$, we have
\begin{gather*}\begin{aligned}
{\rm (Gaussian\;likelihood):}\quad p(\lambda|\yb, \ub, \gamma) & = \Gamma\Big(\alpha_\lambda+\frac{m}2,\; \beta_\lambda + \frac{1}{2}\Vert\Fb(\ub)-\yb\Vert_{\obscov^{-1}}^2\Big), \\
{\rm (Poisson\;likelihood):}\quad p(\lambda|\yb, \ub, \gamma) & = \Gamma\Big(\alpha_\lambda+{\sum_{i = 1}^{m}\data_i},\; \beta_\lambda + \sum_{i = 1}^{m} \forward_i(\param)\Big),
\end{aligned}\end{gather*}
and for updating $\delta$, we have
\begin{gather*}
p(\delta |\yb, \ub, \gamma) = \Gamma\Big(\alpha_\delta+\frac{n}{2},\;\beta_\delta + \frac{1}{2}\big\Vert\ub-\prmean\big\Vert^2_{\prprech}\Big).
\end{gather*}
As a result, we can directly draw samples from the conditional density $p(\lambda, \delta |\yb, \ub, \gamma)$.

\paragraph{3. Updating $\gamma$ given $(\ub, \lambda, \delta)$} For both the Gaussian likelihood and the Poisson likelihood, the conditional distribution $p(\gamma|\yb, \ub, \alpha)$ takes the form
\begin{gather}
\label{eq:gamma_conditional}
p(\gamma|\data,\ub,\lambda,\delta)\propto p_0(\gamma)\,\det \big( \prprech \big)^{\frac12} \exp\Big( -\frac{\delta}{2} \big\Vert\ub-\prmean\big\Vert^2_{\prprech}\Big)
\end{gather}
Since we can not directly sample from $p(\gamma|\yb, \ub, \alpha)$, we present an {\em inverse cumulative distribution function} (inverse CDF) method to explore the conditional distribution.
A key step here is to approximate the function $p(\gamma|\yb, \ub, \alpha)$, in which multiple evaluations of $p(\gamma|\yb, \ub, \alpha)$ is needed.

We exploit the particular structure of the prior precision matrix introduced in Definition \ref{def:prior} to enable the fast evaluation of $p(\gamma|\yb, \ub, \alpha)$. Defining the matrix $\Ab = \bar\Mb^{-1}\Kb$ and denoting its eigenvalues by $\chi_i(\Ab), \ldots, \chi_n(\Ab)$, the determinant in \eqref{eq:gamma_conditional} can be expressed as
\begin{gather*}
\begin{aligned}
(d = 1): \quad \det \big( \prprech \big) & = \det\big(\bar\Mb\big)\, \det\big(\gamma \,\Ib_n + \bar\Mb^{-1} \Kb \big) = \det\big(\bar\Mb\big)\, \prod_{i =1}^n \big(\chi_i(\Ab) + \gamma\big) , \\
(d = 2,3): \quad \det \big( \prprech \big) & = \det\big(\bar\Mb\big)\,\det\big( \gamma^2 \,\Ib_n + 2 \gamma\,\Ab + \Ab^2 \big) = \det\big(\bar\Mb\big)\, \prod_{i =1}^n \big(\chi_i(\Ab) + \gamma\big)^2 .
\end{aligned}
\end{gather*}
Discarding constant terms for fixed $\lambda$, $\delta$, and $\ub$, the conditional distribution can be simplified to
\begin{gather}
\label{eq:gamma_conditional_d1}
p(\gamma|\data,\ub,\lambda,\delta) \propto p_0(\gamma)\,\exp\Big( \frac12\sum_{i =1}^n \log\big(\chi_i(\Ab) + \gamma\big) \,-\frac{\delta\gamma}{2} \big\Vert\ub-\prmean\big\Vert^2_{\bar\Mb} \Big) ,
\end{gather}
for the case $d=1$, and
\begin{gather}
p(\gamma|\data,\ub,\lambda,\delta) \propto p_0(\gamma)\,\exp\Big( \sum_{i = 1}^{n} \log \big(\chi_i(\Ab) + \gamma\big) -\frac{\delta\gamma^2}{2} \big\Vert\ub-\prmean\big\Vert^2_{\bar\Mb} -\gamma\delta \big\Vert\ub-\prmean\big\Vert^2_{\Kb} \Big) ,
\label{eq:gamma_conditional_d2}
\end{gather}
for the cases $d=2,3$. Thus, for any $\ub$ and $(\lambda, \delta)$, the conditional density can be computed at low computational cost---which only needs $\mathcal{O}(n)$ basic arithmetic operations---given the eigenvalues $\chi_i(\Ab)$, for $i = 1, \ldots, n$, are pre-computed before the MCMC simulation.
Then, we can construct the following inverse CDF method can be used to sample from the $p(\gamma|\data,\ub,\lambda,\delta)$. 

\begin{definition}\label{def:inverse_cdf}
Since the hyper-prior random variable $\gamma$ often varies by several order of magnitude, we use the change of variables $\gamma=e^\rho$ to guarantee an accurate approximation of the CDF. This leads to the transformed probability density
\begin{gather}
g(\rho|\data,\ub,\lambda,\delta) = e^\rho \, p(e^\rho|\data,\ub,\lambda,\delta).
\end{gather}
We discretize the interval $\rho \in [\log \gamma_{\rm L}, \log \gamma_{\rm R}]$ using a uniform grid with $n_\rho = 10^3$ grid points.
Then the density $g(\rho|\data,\ub,\lambda,\delta)$ is approximated by a piecewise linear interpolation on this grid, denoted by $\tilde g(\rho|\data,\ub,\lambda,\delta)$.
This way, we obtain a piecewise quadratic approximation to the CDF
\begin{gather}
\tilde G(\rho|\data,\ub,\lambda,\delta) = \int_{\log \gamma_{\rm L}}^{\rho} \tilde g(\rho^\prime|\data,\ub,\lambda,\delta) d\rho^\prime.
\end{gather}
In the inverse CDF method, we draw a random variable $\xi \sim {\rm unifom}(0,1)$, and then compute $\gamma=\exp\big( \tilde G^{-1}(\xi|\data,\ub,\lambda,\delta)\big)$ to obtain a sample from the approximate conditional density
\begin{gather}\label{eq:approx_cond_pdf}
\tilde p(\gamma|\data,\ub,\lambda,\delta) = \frac{1}{\gamma}\,\tilde g(\log \gamma|\data,\ub,\lambda,\delta).
\end{gather}
Then, we can use \eqref{eq:approx_cond_pdf} as a proposal within a MH step to correct for the approximation error.
\end{definition}

\begin{algorithm}[htbp]
\caption{RTO-within-Gibbs for sampling from $p(\ub,\lambda,\delta,\gamma|\data)$} \label{alg:rto_mwh}{}
\begin{algorithmic}[1]
\State Initialize the Markov chain with $\lambda_0,\delta_0,\gamma_0$ and $\ub_0$.
\For {$i=1,\ldots, N$}
\State Find $\ub_\ast$ for $(\lambda_{i-1},\delta_{i-1},\gamma_{i-1})$ and define the RTO density $p_{\rm RTO} (\ub | \lambda_{i-1},\delta_{i-1},\gamma_{i-1})$
\State Set $\ub_i$ = $\ub_{i-1}$ and compute the weight $w(\ub_{i}; \lambda_{i-1},\delta_{i-1},\gamma_{i-1})$
\For {$i=1,\ldots, N_{\rm sub}$}
\State Generate a sample $\ub_\sharp \!\!\sim\! p_{\rm RTO}(\ub| \lambda_{i-1},\delta_{i-1},\gamma_{i-1})$ and compute $w(\ub_\sharp; \lambda_{i-1},\delta_{i-1},\gamma_{i-1})$ 
\State With probability\vspace{-1em}
\begin{gather}
\alpha(\ub_{i},\ub_\sharp)=\min\Big\{1, \frac{w(\ub_\sharp; \lambda_{i-1},\delta_{i-1},\gamma_{i-1})}{w(\ub_{i}; \lambda_{i-1},\delta_{i-1},\gamma_{i-1})} \Big\},\vspace{-1em}
\end{gather}
\quad\quad\;\;\, {\bf accept} $\ub_\sharp$ by setting $\ub_i$ = $\ub_\sharp$.
\EndFor
\State Draw random variables $(\lambda_i,\delta_i)\sim p(\lambda,\delta|\yb,\ub_{i},\gamma_{i-1})$ as defined in Equation \eqref{eq:lambda_delta_conditional}.
\State Draw $\gamma_\sharp \sim \tilde p(\gamma|\yb,\ub_{i},\lambda_i,\delta_i)$ using the inverse CDF method (Definition \ref{def:inverse_cdf}).	
\State With probability\vspace{-1em}
\begin{gather}
\alpha(\gamma_{i-1},\gamma_\sharp)=\min\Big\{1, \frac{p(\gamma_\sharp|\data,\ub_{i},\lambda_i,\delta_i)\,\tilde p(\gamma_{i-1}|\data,\ub_{i},\lambda_i,\delta_i)}{p(\gamma_{i-1}|\data,\ub_{i},\lambda_i,\delta_i)\,\tilde p(\gamma_\sharp|\data,\ub_{i},\lambda_i,\delta_i)} \Big\}%
\end{gather}
\quad\; {\bf accept} $\gamma_\sharp$ by $\gamma_i$ = $\gamma_\sharp$, otherwise {\bf reject} it by by $\gamma_i$ = $\gamma_{i-1}$.
\EndFor
\end{algorithmic}
\end{algorithm}

We call the above procedure of cyclically sampling from the conditional densities the {\it  RTO-within-Gibbs} sampler and summarize it in Algorithm \ref{alg:rto_mwh}. 
In Lines 3--8 of Algorithm \ref{alg:rto_mwh}, we provide the option that taking $N_{\rm sub}$ RTO-MH iterations per RTO-within-Gibbs step to improve the chances of updating the $\ub$-chain. This requires the solution of $N_{\rm sub} + 1$ optimization problems: the first yields the reference parameter $\ub_\ast$ (e.g., the MAP estimator) for $(\lambda_{i-1},\delta_{i-1},\gamma_{i-1})$, and the rest yield $N_{\rm sub}$ RTO-MH iterations. 
In Line 10, given an accurate approximate conditional density $\tilde p(\gamma|\data,\ub,\lambda,\delta)$, the acceptance probability of the $\gamma$-chain can be close to $1$. In such a situation, the behaviour of the MH step here is close to a Gibbs update from the exact conditional $p(\gamma|\data,\ub,\lambda,\delta)$.

\begin{remark}
The inverse CDF method introduced here provides additional modelling flexibilities in the Bayesian inversion, since it can also be applied to sample other hyperparameters when the conditional distribution can not be directly sampled from. For example, we can extend the hyper-prior distributions for $\lambda$ and $\delta$ beyond the current Gamma distribution setting.
\end{remark}

\subsection{Dimension scalability}
Given that Algorithm \ref{alg:rto_mwh} is an extension of the hierarchical Gibbs algorithm of \cite{bardsley2012mcmc} to nonlinear inverse problems, we expect that Algorithm \ref{alg:rto_mwh} exhibits the same dimension scalability issues as the hierarchical Gibbs applied to linear inverse problems. 
As outlined in \cite{agapiou2014analysis}, the correlation in the $\delta$-chain increases as $n\rightarrow\infty$. The exact nature of the dependence between the $\delta$-chain and $n$ is the subject of \cite[Theorem 3.4]{agapiou2014analysis}, where under reasonable assumptions the expected jump size of the $\delta$-chain scales like $2/n$. Specifically, for any $\delta>0$,
\begin{gather*}
\frac{n}{2}\mathbb{E}\big[\delta_{k+1}-\delta_{k}|\delta_{k}=\delta\big]=(\alpha_\delta+1)\delta
-f_n(\delta;\data)\delta^2+\mathcal{O}(n^{-1/2}),
\end{gather*}
where $\mathbb{E}$ denotes expectation and $f_n(\delta;\bb)$ is bounded uniformly in $n$. Moreover, the variance of the step also scales like $2/n$; for any $\delta>0$,
\begin{gather*}
\frac{n}{2}{\rm Var}\big[\delta_{k+1}-\delta_{k}|\delta_{k}=\delta\big]=2\delta^2+\mathcal{O}(n^{-1/2}).
\end{gather*}
A consequence of these results is that the expected squared jumping distance of the Markov chain for $\delta$ is $\mathcal{O}(1/n)$. Moreover, it is noted that the lag-1 autocorrelation of the $\delta$-chain behaves like $1-c/n$ for some constant $c$, but ${\rm Var}(\delta_{k})=\mathcal{O}(1)$. Hence, the Monte Carlo error associated with $N$ draws in stationarity is $\mathcal{O}(\sqrt{n/N})$. {Thus,} the $\delta$-chain becomes increasing correlated as $n\rightarrow\infty$. We will verify this also occurs in nonlinear cases in our numerical experiments in Sections \ref{sec:elliptic} and \ref{sec:PET}.

%% file: sec_PM_a.tex
\section{RTO-within-pseudo-marginal}
\label{sec:Pseudo_algorithms}

To overcome the dimension scaling limit of RTO-within-Gibbs, we will integrate RTO into the pseudo-marginal MCMC \cite{andrieu2009pseudo} to design a new algorithm that jointly update hyperparameters and parameters together. %
The resulting method is named RTO-PM. %

\subsection{RTO-PM MCMC}
Since we will jointly update all hyperparameters together, we group the hyperparameters as $\thetab = (\lambda, \delta, \gamma)$ and denote the hyper-prior density by $p_0(\thetab)$.
We can either adopt the hyper-prior specifications given in Definition \ref{def:hyper} or use more general definitions, as the method presented here does not rely on the Gibbs update in Section \ref{sec:Gibbs_algorithms}.
For a given $\thetab$, suppose we have a RTO density $p_{\rm RTO}(\param|\thetab)$, defined by (\ref{eq:rto_density}), then the marginal likelihood can be expressed as
\begin{gather*}
\like(\data|\thetab) = \int \frac{f(\param|\data,\thetab)}{p_{\rm RTO}(\param\vert\thetab)} p_{\rm RTO}(\param\vert\thetab)\, d\param = \int w(\param;\thetab) p_{\rm RTO}(\param\vert\thetab)\, d\param,
\end{gather*}
where $w(\param;\thetab)$ is defined by either \eqref{eq:rto_weight} or \eqref{eq:rto_weight_poisson}. Thus, we can use the RTO density and importance sampling to estimate the marginal likelihood. In fact, using the pseudo-marginal principle \cite{andrieu2009pseudo} and the importance sampling formula \eqref{eq:marg_like_PM}, we can derive asymptomatically convergent MCMC methods that have the exact marginal posterior,
\begin{gather*}
p(\thetab|\data)=\frac{1}{p(\data)} \like(\data|\thetab) p_0(\thetab),
\end{gather*}
as the invariant density, and simultaneously sample from the joint posterior $p(\param, \thetab|\data)$.

\begin{definition}[pseudo-marginal density]\label{def:PM}
We define a joint importance sampling density
\begin{gather}\label{eq:joint_proposal}
g(\mathcal{U} | \thetab) = \prod_{i = 1}^{K} p_{\rm RTO}(\param^{i}\vert\thetab), \quad \textrm{where} \quad \mathcal{U} = \{\ub^{1}, \ldots, \ub^{K}\},
\end{gather}
Then, drawing a set of random variables $\mathcal{U}$ from $g(\mathcal{U} | \thetab)$, we can compute
\begin{gather*}\label{eq:PM_like}
\like_{K}(\data|\thetab)  \coloneqq \frac{1}{K}\sum_{i = 1}^{K} \frac{f(\param^{i}|\data,\thetab)}{p_{\rm RTO}(\param^{i}\vert\thetab)} = \frac{1}{K}\sum_{i = 1}^{K} w(\param^{i};\thetab),
\end{gather*}
which is the estimator of the marginal likelihood $\like(\data|\thetab)$. This defines the pseudo-marginal density
\begin{gather}
\label{eq:PM}
p_{K}(\thetab|\data)\propto \like_{K}(\data|\thetab)p_0(\thetab).
\end{gather}
\end{definition}

Following the derivation in \cite{andrieu2009pseudo}, the product of the pseudo-marginal density \eqref{eq:PM} and the importance density \eqref{eq:joint_proposal} defines a joint density in the form of
\begin{gather}\label{eq:PM_joint}
p(\thetab, \mathcal{U}) = \frac{1}{p(\data)} \,g(\mathcal{U} | \thetab) \; p_{K}(\thetab|\data)= \frac{p_0(\thetab)}{p(\data)}\, g(\mathcal{U} | \thetab) \; \frac{1}{K}\sum_{i = 1}^{K} w(\param^i;\thetab).
\end{gather}
Marginalizing the joint density $p(\thetab, \mathcal{U})$ over $\mathcal{U}$, we obtain the marginal posterior:
\begin{gather}\begin{aligned}
\int p(\thetab, \mathcal{U}) d\mathcal{U}
& = \frac{p_0(\thetab)}{p(\data)}\,\frac{1}{K}\sum_{i = 1}^{K} \int \frac{f(\param^{i}|\data,\thetab)}{p_{\rm RTO}(\param^{i}\vert\thetab)} \bigg(\int \prod_{j \neq i} p_{\rm RTO}(\param^{j}\vert\thetab) d\ub^j \bigg) p_{\rm RTO}(\param^{i}\vert\thetab) d\ub^{i}\\
& = \frac{1}{p(\data)}\like(\data|\thetab)\,p_0(\thetab).
\end{aligned}
\end{gather}

\begin{proposition}\label{prop:PMMCMC}
Consider that we have a MH method drawing proposal candidates from a distribution $\thetab_\sharp \sim q(\cdot | \thetab )$ and accepting the proposal candidate with the probability
\begin{gather}\label{eq:PM_acc}
\alpha_K(\thetab, \thetab_\sharp) = \min\bigg\{ 1, \frac{p_{K}(\thetab_\sharp|\data) \,q(\thetab | \thetab_\sharp )}{p_{K}(\thetab|\data)\,q(\thetab_\sharp | \thetab )} \bigg\}.
\end{gather}
It constructs an ergodic Markov chain with the marginal posterior $p(\thetab|\data)$ as the invariant density.
\end{proposition}

\begin{proof}
Since we have the joint distribution $p(\thetab, \mathcal{U})=p(\data)^{-1}\like_{K}(\data|\thetab)\,p_0(\thetab)\, g(\mathcal{U}| \thetab)$ and $p(\data)$ is a constant, the Metropolis-Hastings ratio in the acceptance probability takes the form
\begin{gather*}
\frac{p_{K}(\thetab_\sharp|\data) \,q(\thetab | \thetab_\sharp )}{p_{K}(\thetab|\data)\,q(\thetab_\sharp | \thetab )}
= \frac{\like_K(\data|\thetab_\sharp) \,p_0(\thetab_\sharp)\,q(\thetab | \thetab_\sharp )}{\like_K(\data|\thetab) \,p_0(\thetab)\,q(\thetab_\sharp | \thetab )}  = \frac{p(\thetab_\sharp, \mathcal{U}_{\sharp})\,g(\mathcal{U}| \thetab) \,q(\thetab | \thetab_\sharp )}{p(\thetab, \mathcal{U})\,g(\mathcal{U}_{\sharp}| \thetab_\sharp)\,q(\thetab_\sharp | \thetab )}
\end{gather*}
This effectively defines a MH method that samples the joint distribution $p(\thetab, \mathcal{U})$ using the proposal $g(\mathcal{U}_{\sharp}| \thetab_\sharp)\,q(\thetab_\sharp | \thetab )$.
Since the joint distribution $p(\thetab, \mathcal{U})$ has $p(\thetab|\data)$ as its marginal, the result follows.
\end{proof}

The resulting pseudo-marginal method for sampling from $p(\thetab|\data)$ is given in Algorithm \ref{alg:rto_pm}. 
Note that we can optionally save the RTO sample set $\mathcal{U}_i$ and weight set $\mathcal{W}_i$ (Lines 10 and 12) to use them in importance sampling for estimating expectations over the joint posterior and draw parameter samples from the joint posterior (Line 13). In the latter case, one can randomly draw a sample $\ub_i$ from the set $\mathcal{U}_i$ according to the categorical distribution defined by the weights $\mathcal{W}_i$.

\begin{algorithm}[htbp]
\caption{RTO-pseudo-marginal for sampling from $p(\thetab|\data)$} \label{alg:rto_pm}{}
\begin{algorithmic}[1]
\State Initialize $\thetab_0$, find a corresponding reference parameter $\ub_\ast$, and define $p_{\rm RTO} (\ub | \thetab_0)$
\State Compute a RTO sample set $\mathcal{U}_0\equiv\{\ub_0^{j}\}_{j=1}^K$ and the weight set $\mathcal{W}_0\equiv \{w(\ub_0^{j}; \thetab_0)\}_{j = 1}^{K}$
\State Evaluate pseudo-marginal density $\pi_K(\thetab_0| \data)$ using (\ref{eq:PM})
\For {$i=1,\ldots, N$}
    \State Draw a proposal candidate $\thetab_\sharp \sim q(\cdot|\thetab_{i-1})$
    \State Find a reference parameter $\ub_\ast$ for $\thetab_\sharp$ and define the RTO density $p_{\rm RTO} (\ub | \thetab_\sharp)$
	\State Compute a RTO sample set $\mathcal{U}_\sharp =\{\ub^{j}_\sharp\}_{j=1}^K$ and the weight set $\mathcal{W}_\sharp = \{w(\ub_\sharp^{j}; \thetab_\sharp)\}_{j = 1}^{K}$
	\State Evaluate the pseudo-marginal density $p_K(\thetab_\sharp|\yb)$ using (\ref{eq:PM})
	\State With probability $\alpha_K(\thetab_{i-1},\thetab_\sharp)$ defined in \eqref{eq:PM_acc}, {\bf accept} by setting $\thetab_i = \thetab_\sharp$
	\State \quad \quad {\bf Optional:} set $\mathcal{U}_i= \mathcal{U}_\sharp$ and $\mathcal{W}_i = \mathcal{W}_\sharp$,
	\State Otherwise {\bf reject} by setting $\thetab_i = \thetab_{i-1}$
	\State \quad \quad {\bf Optional:} set $\mathcal{U}_i= \mathcal{U}_{i-1}$ and $\mathcal{W}_i = \mathcal{W}_{i-1}$
	\State {\bf Optional:} Draw $\ub_i \in \mathcal{U}_i$ according to the categorical distribution defined by $\mathcal{W}_i$
\EndFor
\end{algorithmic}
\end{algorithm}

\subsection{Computational remarks}\label{remark:pm_comp}
The generation of the RTO samples per RTO-PM step (in Lines 6--8) requires the solution of $K + 1$ optimization problems: the first yields the reference parameter $\ub_\ast$ (e.g., the MAP estimator) for $\thetab_{i-1}$, and the rest are used for computing the pseudo-marginal density.
In our implementation of RTO-PM, we use the adaptive Metropolis proposal distribution \cite{MCMC:HST_2001} in Line 5. Moreover, in RTO-PM, we do not have to compute the determinant of the prior precision matrix (for updating $\gamma$) as we did in the Gibbs case. Instead, we only need to compute the determinants of $r \times r$ matrices in either Equation \eqref{eq:rto_weight} or Equation \eqref{eq:rto_weight_poisson}.
The pseudo-marginal density $p_{K}(\thetab|\data)$ is a random variable, in which a larger sample size $K$ results in an estimate of the marginal density $p(\thetab|\data)$ with a smaller variance compared to that obtained using a smaller $K$.
Thus, using a larger $K$ may result in better statistical efficiency for exploring hyperparameters---which can be measured by the {\em integrated autocorrelation time} (IACT), see Remark \ref{remark:iact} below---compared with that of a smaller $K$. 
However, the computing cost of $p_{K}(\thetab|\data)$ increases linearly with $K$, while the improvement of the statistical efficiency will not follow the same rate. 
We refer the readers to \cite{andrieu2016establishing,doucet2015efficient} for a detailed discussion on this topic and only provide an interpretation as follows. 
Using the ``oracle'' acceptance probability
\begin{gather}\label{eq:M_acc}
\bar\alpha(\thetab, \thetab_\sharp) = \min\bigg\{ 1, \frac{p(\thetab_\sharp|\data) \,q(\thetab | \thetab_\sharp )}{p(\thetab|\data)\,q(\thetab_\sharp | \thetab )} \bigg\},
\end{gather}
we can define a standard MCMC method targeting the marginal density $p(\thetab|\data)$.
Since the RTO-PM acceptance probability $\alpha_K(\thetab, \thetab_\sharp) \rightarrow \bar\alpha(\thetab, \thetab_\sharp)$ as $K \rightarrow \infty$, the statistical efficiency of RTO-PM will reach that of the standard MCMC and cannot improved further. 
This way, for sufficiently large $K$, we expect that the error in the pseudo-marginal density will has negligible impact on the statistical efficiency for exploring $\thetab$. 
Under restrictive assumptions, the result of \cite{doucet2015efficient} offers a rule-of-thumb for choosing the sample size: the sample size $K$ can be chosen such that the standard deviation of the log-pseudo-marginal density, $\text{var}[\log p_{K}(\thetab|\data)]^{\frac12}$, is approximately $0.92$. 

\begin{remark}\label{remark:iact} The integrated autocorrelation time is computed from the autocorrelation function (ACF) of a Markov chain. For a Markov chain $\{\delta_k\}_{k=1}^N$, the ACF is estimated as
\begin{gather}
\label{ACF}
\hat\rho(j)= C(j)/C(0),
\quad{\rm where}\quad
C(j)=\frac{1}{N-j}\sum_{k=1}^{N-j} (\delta_k-\bar\delta)(\delta_{k+|j|}-\bar\delta),
\end{gather}
where $\bar\delta=\frac1N\sum_{k=1}^N\delta_k$. The faster $\hat\rho(j)$ decays to zero, the less correlated is the $\delta$-chain.  The IACT of $\{\delta_k\}_{k=1}^N$, denoted by $\tau_{\rm int}(\delta)$, is estimated as the summation of the truncated ACF (see \cite{bardsley2018book,liu2008monte} for details). The faster (slower) $\hat\rho(j)$ decays to zero, the smaller (larger) will be the IACT. %
\end{remark}

\begin{remark}%
With $K = 1$, the acceptance probability of RTO-PM can be expressed as
\begin{gather*}
\begin{aligned}
\alpha(\thetab,\thetab_\sharp)
&=\min\bigg\{1,\frac{w(\param_\sharp;\thetab_\sharp)\,p_0(\thetab_\sharp)\,q(\thetab|\thetab_\sharp)}{w(\param_{i-1};\thetab)\,p_0(\thetab)\,q(\thetab_\sharp|\thetab)}\bigg\} \\
&=\min\bigg\{1,\frac{f(\param_\sharp|\yb,\thetab_\sharp)\,p_0(\thetab_\sharp)\,p_{\rm RTO}(\ub_{i-1}|\thetab)\,q(\thetab|\thetab_\sharp)}{f(\ub_{i-1}|\yb,\thetab)\,p_0(\thetab)\,p_{\rm RTO}(\ub_\sharp|\thetab_\sharp)\,q(\thetab_\sharp|\thetab)}\bigg\} \\
&= \min\bigg\{1,\frac{p(\param_\sharp,\thetab_\sharp|\yb)\,p_{\rm RTO}(\ub_{i-1}|\thetab)\,q(\thetab|\thetab_\sharp)}{p(\ub_{i-1},\thetab|\yb)\,p_{\rm RTO}(\ub_\sharp|\thetab_\sharp)\,q(\thetab_\sharp,\thetab)}\bigg\}.
\end{aligned}
\end{gather*}
This way, the RTO-PM is equivalent to a MH method that uses a proposal $p_{\rm RTO}(\ub_\sharp|\thetab_\sharp)\,q(\thetab_\sharp,\thetab)$ to sample the joint posterior $p(\param,\thetab|\yb)$.
Thus, for $K =1$, the RTO-PM can also be viewed as the nonlinear extension of the one-block-update (see \cite {fox2016fast,rue2005gaussian,saibaba2019} for example) used in linear inverse problems.
In the linear case, the conditional posterior $p(\param|\yb,\thetab)$ is Gaussian and can be directly sampled, whereas here we sample from $p(\param|\yb,\thetab)$ using an MH step with RTO proposal.
\end{remark}

\subsection{Dimension scalability}
\label{remark:pm_dim}
Algorithm \ref{eq:PM} is the nonlinear analogue of sampling from the marginal density for linear inverse problems, as found in \cite[Section 5.3]{bardsley2018book} and \cite{fox2016fast,saibaba2019}, where it is shown that sampling from the marginal density removes the dimension scalability issues for the Gibbs sampler described at the end of Section \ref{sec:Gibbs_algorithms}. We expect the same result in the nonlinear case when using RTO-PM, since we do not simulate a Markov chain in the high-dimensional model parameter ($\ub$) space. The estimate of the marginal likelihood depends on the weight, which has an infinite dimensional limit, as shown in \cite {wang2019scalable}. So we expect that the RTO-PM is robust to model parameter dimension, i.e., its sampling performance should not deteriorate with increasing model parameter dimension. We will verify this in our numerical experiments in Sections \ref{sec:elliptic} and \ref{sec:PET}.

%% file: sec_Numerics_Conclusion_a.tex
\section{Example 1: elliptic inverse problem}
\label{sec:elliptic}

The first example is on the one dimensional PDE-constrained inverse problem with Gaussian measurement noise. 

\subsection{Setup and inversion results}
We aim to estimate the log-diffusion coefficient $u(s)$ from measurements of the potential function $x(s)$ of the Poisson equation
\begin{gather}
\label{HeatIP}
-\frac{d}{ds}\Big( \exp\big(u(s)\big) \frac{d x}{ds}\Big)=f(s),\quad s\in \Omega := [0,1], 
\end{gather}
with boundary conditions $x(0)=x(1)=0$. After numerical discretization, this yields 
\begin{gather}
\label{HeatIPdisc}
\Bb(\ub)\xb=\fb,
\end{gather}
where $\ub \in \R^n$, $\xb\in \R^n$, and $\fb\in \R^n$ are discretizations of $u$, $x$, and $f$; and $\Bb \in \R^{n \times n}$ is the stiffness matrix with imposed zero Dirichlet boundary conditions. 

We adopt the setup presented in \cite{bardsley2018book}. The measurements of the potential function are taken at $63$ equally spaced discrete locations in $[0,1]$ and the observation operator $\Hb \in \R^{63 \times n}$ is used to map the discretized potential function $\xb$ to observables. 
We generate two data sets corresponding to two scaled Dirac delta forcing functions:
\begin{gather*}
f_1(s)=1000\cdot\delta(s-1/3) \quad{\rm and}\quad f_2(s)=1000\cdot\delta(s-2/3),
\end{gather*}
and denote their discretized versions by $\fb_1$ and $\fb_2$, respectively.
Assuming the measurements are corrupted by zero-mean i.i.d. Gaussian noise, the measurement process can be written as
\begin{gather}\label{eq:testcase1}
\data \sim \N(\predata, \lambda^{-1} \obscov),  \quad {\rm subject\;to} \quad \predata=\forward(\param),
\end{gather}
where
\begin{gather}\nonumber
\Fb(\ub):=\left[\begin{array}{c}\Hb\,\Bb(\ub)^{-1}\fb_1 \\ \Hb\,\Bb(\ub)^{-1}\fb_2\end{array}\right], \quad \obscov = \Ib_m, \quad {\rm and}\quad \data, \predata \in \R^{m},  
\end{gather}
with $m = 126$. We generate synthetic data using \eqref{eq:testcase1} with the ``true'' log-diffusion coefficient
\begin{gather*}
u_{\rm true}(s)= \min\left\{1,1 - 0.5\,\sin(2\pi (s-0.25))\right\},
\end{gather*}
discretized forward model with  $n=8192$, and $\lambda^{-1}$ corresponding to a signal-to-noise ratio of 100, i.e., 1\%  noise. 
The corresponding data vectors are plotted in the top left plot of Figure \ref{fig:ellptic} together with the noise-free data $\Hb\,\Bb(\ub)^{-1}\fb_1$ and $\Hb\,\Bb(\ub)^{-1}\fb_2$. 

For solving the inverse problem, we employ the prior defined by the Laplace-like stochastic partial differential equation in Section \ref{sec:prior} with the precision matrix $\prprech = (\gamma \,\bar\Mb + \Kb )$. We employ the hyper-prior distribution specified in Definition \ref{def:hyper}. For $p_0(\lambda)$ and $p_0(\delta)$, we set $\alpha_{\lambda}=\alpha_{\delta}=1$ and $\beta_{\lambda}=\beta_{\delta}=10^{-4}$, which have been shown to be effective on a variety of test cases. For $p_0(\gamma)$, we choose, $\alpha_\gamma = 0$, $\beta_\gamma = 4$, $\gamma_{\rm L} = 10^{-5}$ and $\gamma_{\rm R} = 10$.
We setup three scenarios with $n = 256,1024$, and $4096$ to test the dimension scalability of our proposed methods. 
\begin{figure}[h!]
\centering
\includegraphics[trim = 1.6em 2em 0em 1em , width = 0.4\textwidth]{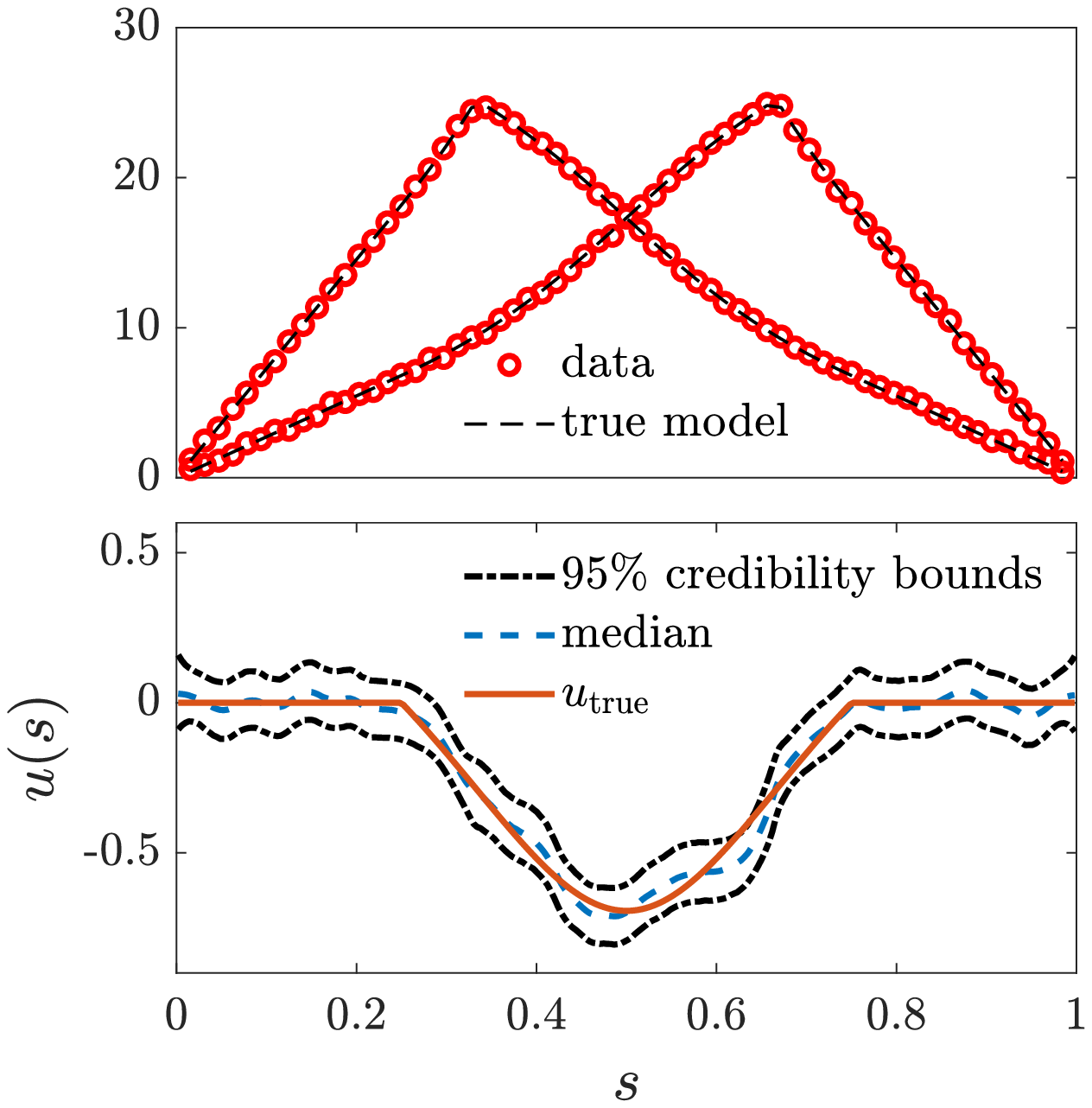}
\includegraphics[trim = 2em 2.3em 0em 1em , width = 0.4\textwidth]{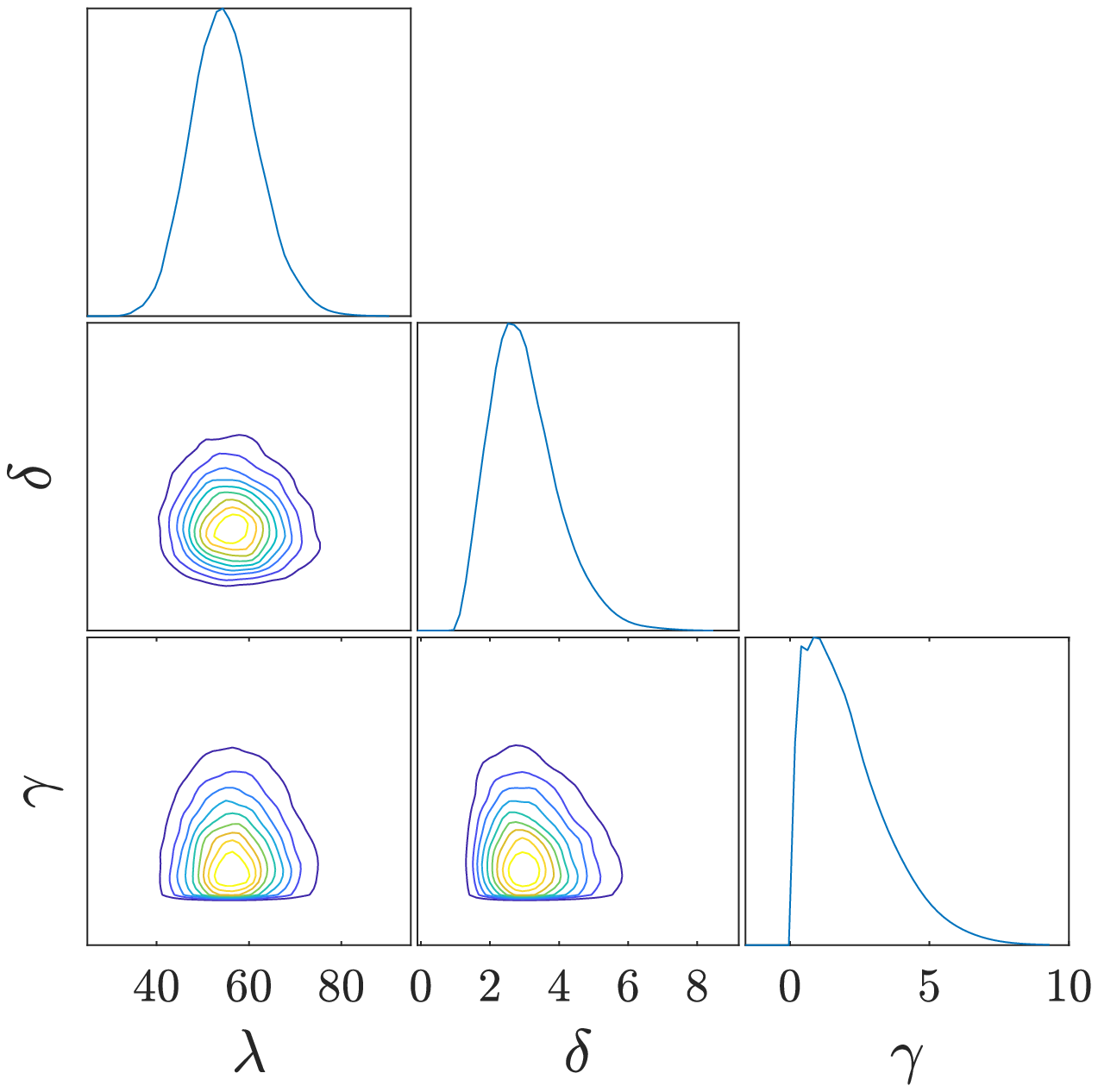}
\caption{Top left: the measured data for the elliptic PDE problem. The remaining plots were computed using RTO-within-Gibbs with $n = 256$. Bottom left: the median and 95\% credibility bounds of $\ub$-samples. Right: the estimated marginal distributions and the pairwise marginal distributions for the hyperparameters.}\label{fig:ellptic}
\end{figure}
Before discussing the numerical experiments, we present the inversion results obtained using RTO-within-Gibbs with the scenario $n = 256$ in Figure \ref{fig:ellptic}: the bottom left plot shows the sample median and 95\% credibility bounds of log-diffusion coefficient; and the right plot shows the marginal densities and the pairwise marginal densities for these hyper-parameters. 
The median, 95\% credibility bounds, and plots of marginal densities obtained for other scenarios and sample methods are similar to those of $n = 256$, and hence are not reported for brevity.

\subsection{Sampling performance}
We first discuss the sampling performance of RTO-within-Gibbs (Algorithm \ref{alg:rto_mwh}) with $N_{\rm sub} = 1$.
\begin{figure}[htbp]
\centering
\pic{.8}{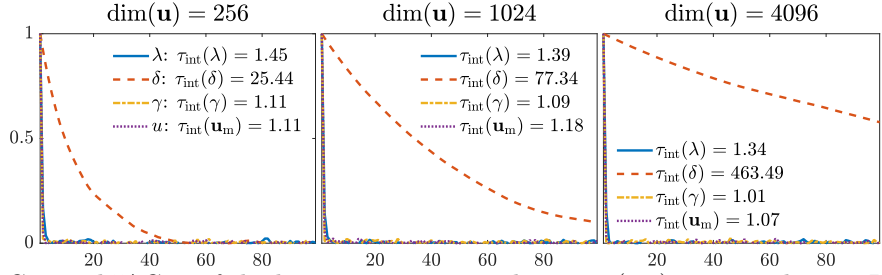}\vspace{-1em}
\caption{ACFs and IACTs of the hyperparameters and $\ub_{\rm m} = u(0.5)$ computed using RTO-within-Gibbs with $n=256, 1024$, and $4096$.}\label{fig:ellptic_gibbs}
\end{figure}
We compute chains of length 18,000 (after burn-in has been removed) for the three scenarios and report the ACFs and IACTs of the hyperparameter chains and the $\ub_{\rm m}$-chain (which is the Markov chain of the model parameter at $s = 0.5$) in Figure \ref{fig:ellptic_gibbs}. 
We observed that the sampling performance of the RTO-within-Gibbs deteriorates as the parameter dimension increases---with increasing $n$, the IACT of the $\delta$-chain increases at roughly the same rate---as we expected in Section \ref{sec:Gibbs_algorithms}. 
We also observed that the IACTs of the $\lambda$-chain, $\gamma$-chain, and the $\ub_{\rm m}$-chain do not change with the parameter dimensions. 
For all three scenarios, the acceptance rate of the RTO sampling (Lines 4--8 in Algorithm \ref{alg:rto_mwh}) has an is about $96\%$.

In our second experiment, we apply the RTO-PM (Algorithm \ref{alg:rto_pm}) with the sample sizes $K=1$ and $K=5$ in the computation of the pseudo-marginal density \eqref{eq:PM}.   
The ACFs and IACTs of the resulting Markov chains are reported in the top row ($K=1$) and the bottom row ($K=5$) of Figure \ref{fig:ellptic_pm}. 
\begin{figure}[htbp]
\centering
\includegraphics[trim = 0em 1.7em 0em 1em , width = 0.8\textwidth]{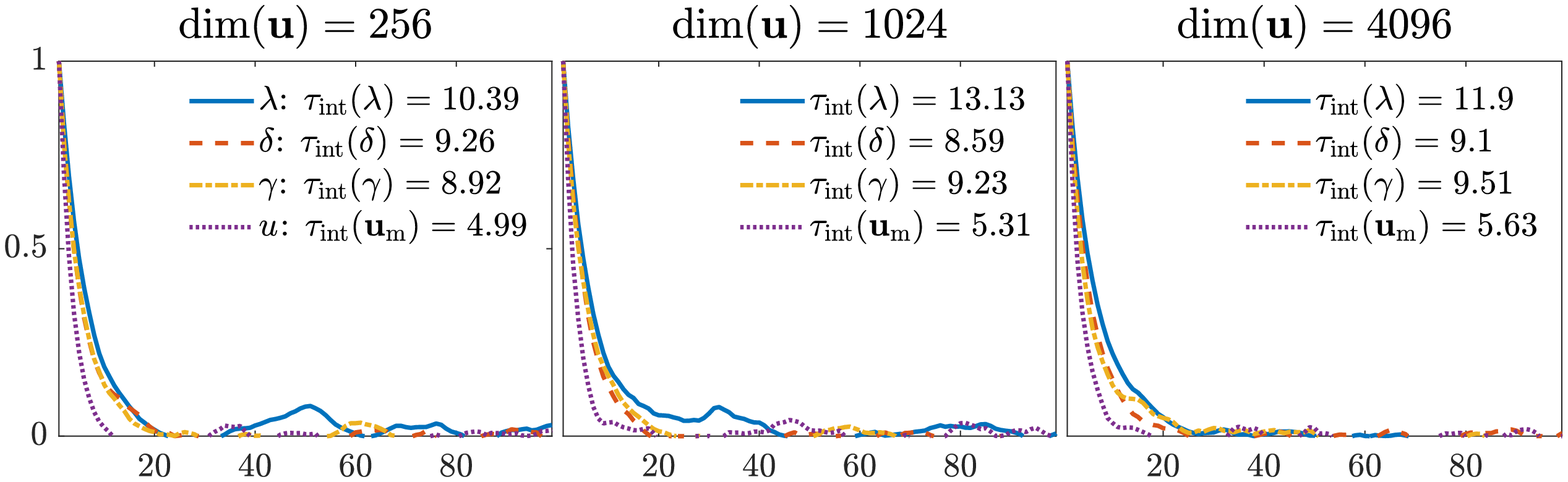}\\
\includegraphics[trim = 0em 0em 0em 3em,clip, width = 0.8\textwidth]{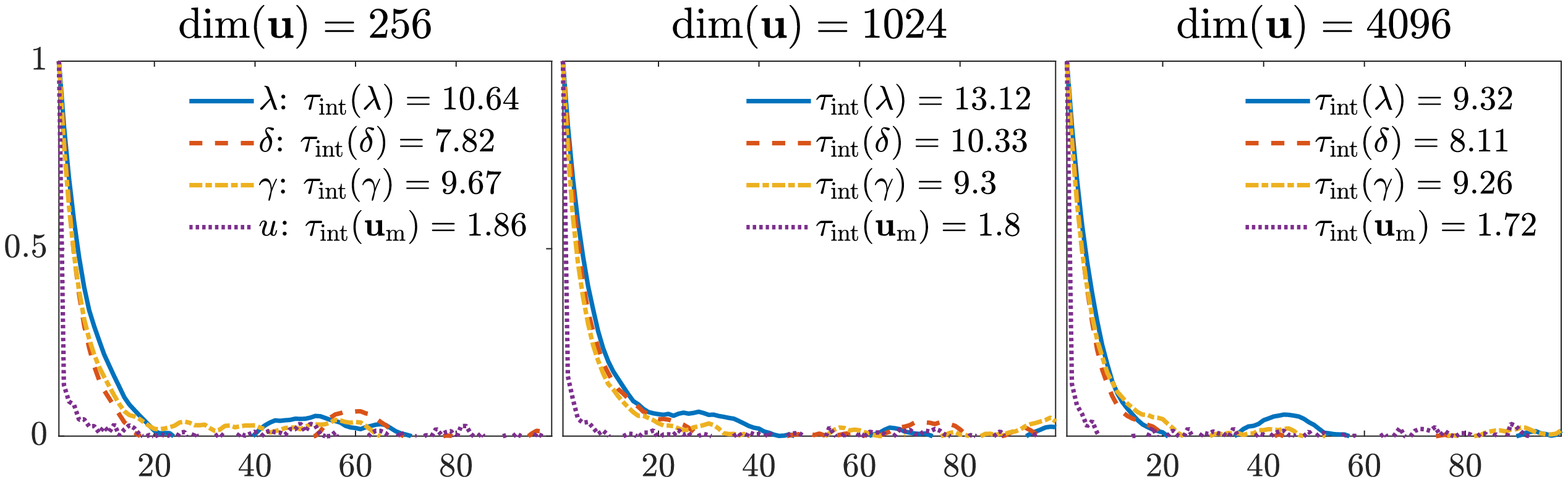}\vspace{-0.5em}
\caption{ACFs and IACTs of the hyperparameters and $\ub_{\rm m} = u(0.5)$ computed using RTO-PM with $K=1$ (top row) and $K=5$ (bottom row) for $n=256, 1024$, and $4096$.}\label{fig:ellptic_pm}
\end{figure}
Note that the case $K=1$ is the nonlinear extension of the one-block-update \cite{rue2005gaussian}. 
In both cases of this experiment, we observed that the Markov chains generated by RTO-PM do not exhibit the dimension scalability issues, where IACTs of the hyper-parameters stabilizes as $n$ increases. This result agrees with our expectation noted in Section \ref{remark:pm_dim}. 
Furthermore, the computational cost per RTO-MH step (with $K=1$) is the same as the computational cost per RTO-within-Gibbs step (with $N_{\rm sub} = 1$). Thus, it is clear that RTO-PM is both statistically and computationally more efficient than RTO-within-Gibbs in this example.

\begin{figure}[h!]
\centering
\pic{.4}{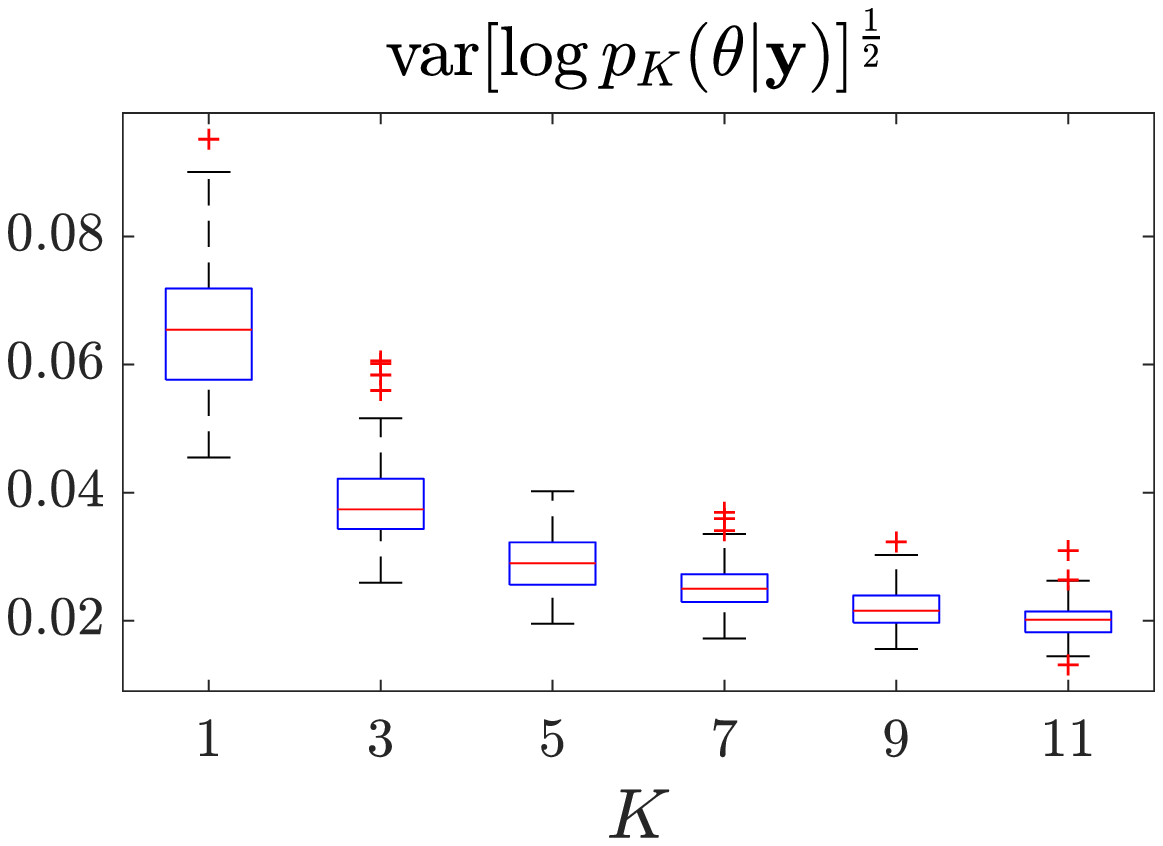}\vspace{-1em}
\caption{The box plot of $\text{var}\big[\log p_{K}(\thetab|\data)\big]^{\frac12}$ for various $K$ (horizontal axis) and posterior hyperparameters (vertical axis).}\label{fig:ellptic_var_pm}
\end{figure}

In the second experiment, we also observed that the ACFs of the hyperparameter chains computed by the one-block-update (RTO-PM with $K=1$) is similar to those computed by RTO-PM with $K=5$. 
To analyze the impact of the sample size $K$, we estimated the standard deviation of the log-pseudo-marginal density, $\text{var}[\log p_{K}(\thetab|\data)]^{\frac12}$, for various sample sizes and various posterior samples of $\thetab$. 
The box plot of estimated values are shown in Figure \ref{fig:ellptic_var_pm}, in which $\text{var}[\log p_{K}(\thetab|\data)]^{\frac12}$ is below $0.1$ for all $K$ and $\thetab$. 
This suggests that, in this example, the RTO density is a close approximation to the conditional posterior density, and using $K=1$ is sufficient for exploring the hyperparameters (see Section \ref{remark:pm_comp}). 

In contrast to the similar ACFs of the hyperparameter chains, the IACTs of the $\ub_{\rm m}$-chains obtained by $K=5$ is approximately three times lower than that of $K=1$. 
This suggests that although using a larger $K$ may not further improve the sampling efficiency of the hyperparameters, it can be used for better exploring the model parameter space. 
However, to the authors' knowledge, the impact of sample size $K$ on the efficiency of exploring model parameters has not been investigated. This could be a potential future research topic.

\section{Example 2: PET imaging}
\label{sec:PET}
The second example is a two dimensional PET imaging problem with Poisson observed data. 
In PET imaging, gamma rays travel from sources to detectors through an object of interest, and then the detectors measure the intensities of gamma rays from multiple sources via the counting of incident photons.
Given the domain of interest $\Omega$, we aim to reconstruct the density of the object
from the counting data recorded at the detectors. 

\subsection{Setup and inversion results}
We assume the unknown density of the object is positive and follows a log-normal prior distribution. This way, we have the density function represented as 
$\exp(u(s))$ for $s \in \Omega$, where $u(s)$ is unbounded and follows a Gaussian prior. 
Then, the change of intensity of an gamma ray along the path, $\ell_i(s), s \in \Omega$, can be modelled using Beer's law:
\begin{gather}\label{eq:beers}
I_{d, i}  = I_{s,i} \exp \Big( - \int_{\ell_{i}(s)} \exp \big( u(s) \big)  d s \Big),
\end{gather}
where $I_{d,i} \in \R_{\geq 0}$ and $I_{s,i} \in \R_{\geq 0}$ are the intensities at the detector and at the source, respectively.
We assume that all the gamma ray sources have the same intensity, $I_{s,i} = \lambda$ for $i = 1, \ldots, m$, in which $\lambda$ is also unknown and will be estimated in the inverse problem as a hyperparameter. 

In this case, the domain $\Omega$ is discretized into a regular grid with $n$ cells and the logarithm of the density is assumed to be piecewise constant. 
This yields the discretized parameter $\ub \in \R^n$. %
The line integrals in \eqref{eq:beers} are approximated by
\begin{gather*}
\int_{\ell_{i}(s)} \exp ( u(s) )  d s \approx \sum_{j = 1}^{n} \Bb_{ij} \, \exp(\ub_j),
\end{gather*}
where $\Bb_{ij} \in \R_{\geq 0}$ is the length of the intersection between line $\ell_{i}$ and cell $j$, and $\exp(\ub_j)$ is the discretized density in cell $j$.
Suppose we have a total of $m$ number of gamma ray paths and the corresponding matrix $\Bb \in \R^{m \times n}$, 
the forward model $\forward: \R^{n} \rightarrow \R^{m}$ can be defined as
\begin{gather}
\forward(\ub) := \exp( - \Bb \, \exp(\ub) ).
\label{eq:forward_xray}
\end{gather}
Since the matrix $\Bb$ has non-negative entries, the forward model outputs are non-negative and bounded, and thus Condition 3 of the Assumption \ref{assum:assum1} is satisfied. 
The detectors record integer-valued counting data of incident photons, $\data \in \mathbb{N}^{m}$. The counting data are modelled by Poisson distributions, in which the expected counts are defined by the scaled forward model $\lambda \forward(\ub)$. 
This way, the Poisson likelihood introduced in Definition \ref{eq:poisson_like} is used to describe the measurement process.

\begin{figure}[!htbp]
\centering
\includegraphics[trim = 0em 0.5em 0em 2em, clip, width = 0.4\textwidth]{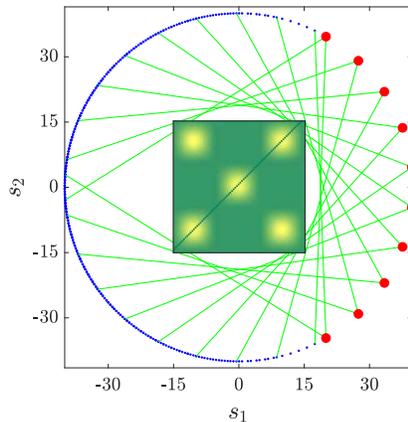}\vspace{-1em}
\caption{The PET setup. The ``true'' log-density function used for generating synthetic data is shown in the domain of interest $\Omega = [-15, 15]^2$. Red dots and blue dots indicate radiation sources and detectors, respectively. 
The black line along the diagonal indicates locations we plot the credible intervals and median estimates for model parameters. 
}\label{fig:pet_setup}
\end{figure}

We consider a PET setup shown in Figure \ref{fig:pet_setup}: the problem domain is $\Omega = [-15, 15]^2$, 10 radiation sources are positioned with equal spaces on one side of a circle, spanning a $120^{\circ}$ angle, and each source sends a fan of 40 gamma rays that are measured by detectors. 
The model setup is computed using the code of \cite{heikkinen2008statistical}. We generate synthetic data using the ``true'' log-density function 
\begin{gather*}
u_{\rm true}(s)= \max\{0, 0.5\, \pi \,\sin\big(0.1\,\pi\,(s_1 - 15))\,\sin\big(0.1\,\pi\,(s_2 - 15)) \big\}, 
\end{gather*}
which is shown in Figure \ref{fig:pet_setup}. 
In the inverse problem, we employ the hyper-prior distribution specified in Definition \ref{def:hyper}. For $p_0(\lambda)$ and $p_0(\delta)$, we set $\alpha_{\lambda}=\alpha_{\delta}=1$ and $\beta_{\lambda}=\beta_{\delta}=10^{-4}$. For $p_0(\gamma)$, we choose, $\alpha_\gamma = 0$, $\beta_\gamma = 4$, $\gamma_{\rm L} = 10^{-3}$ and $\gamma_{\rm R} = 10^2$.

\begin{figure}[!htbp]
\centering
\includegraphics[trim = 2em 1em 0em 0em , width = 0.4\textwidth]{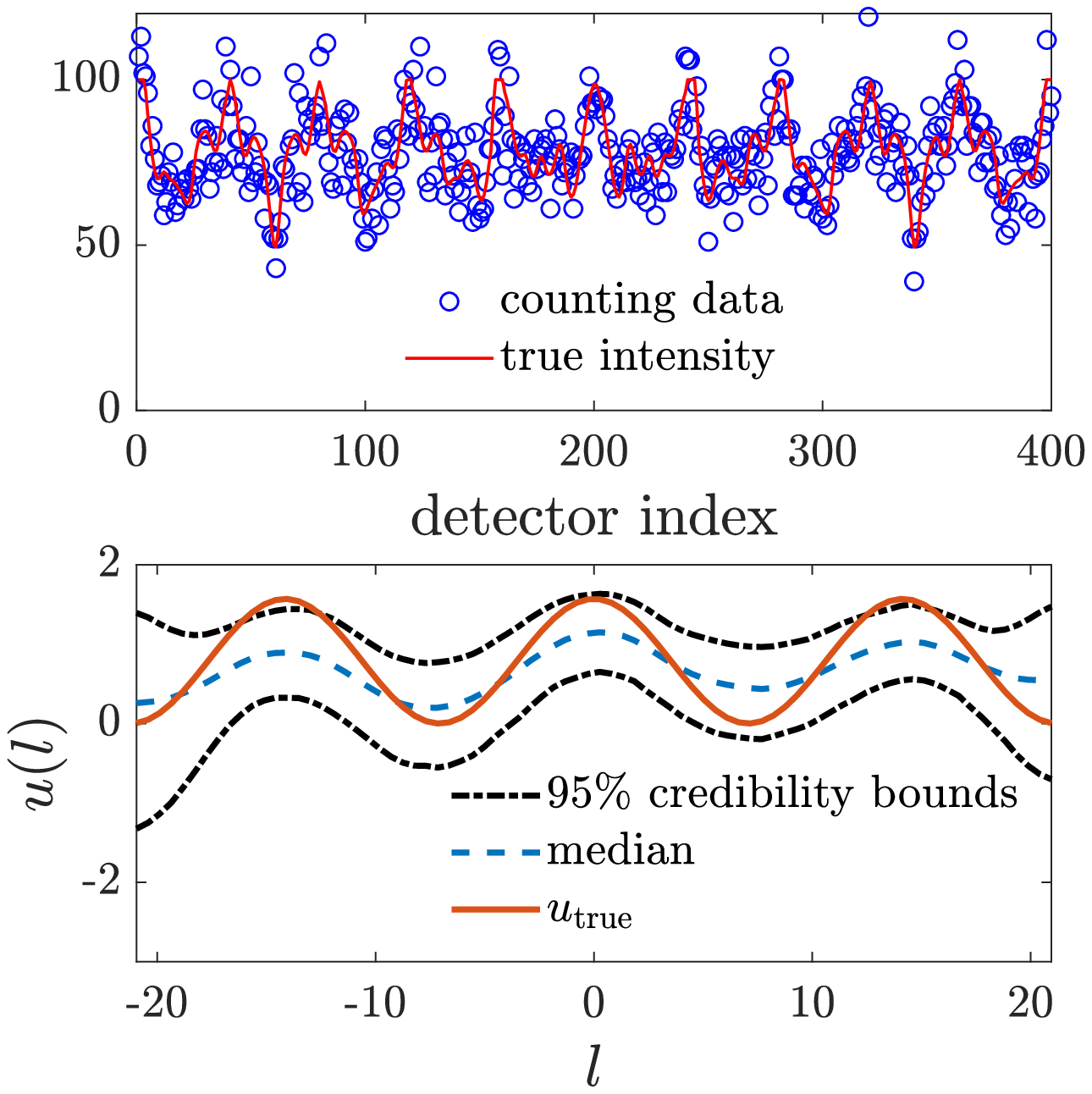}
\includegraphics[trim = 2em 2.3em 0em 1em , width = 0.4\textwidth]{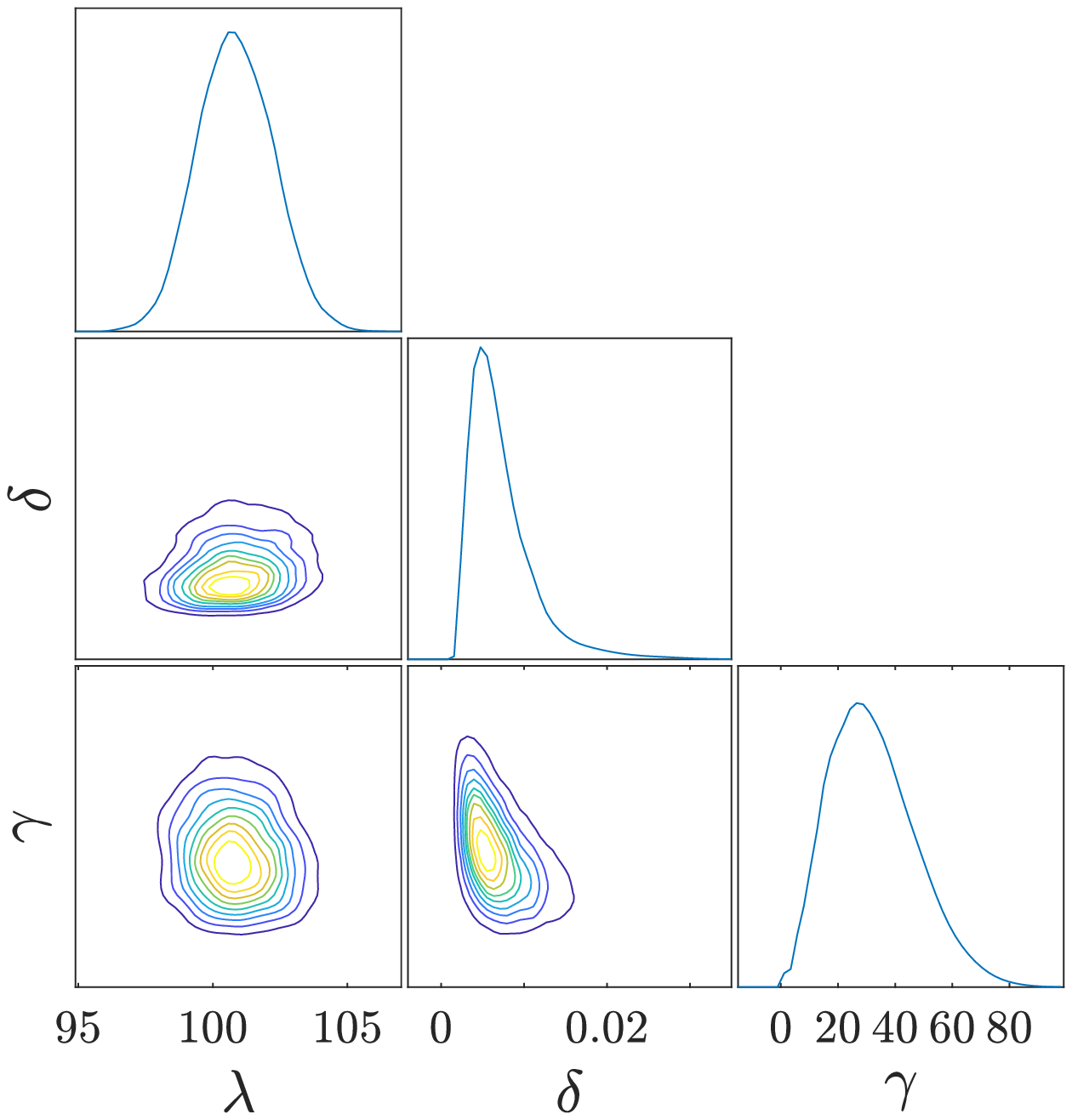}
\caption{
Top left: the measured counting data and the ``true'' intensity profile. The remaining plots were computed using RTO-PM with $n = 6400$. Bottom left: the median and 95\% credibility bounds of the log-density function estimated along the diagonal. Right: the estimated marginal distributions and the pairwise marginal distributions for the hyperparameters.}\label{fig:pet_result}
\centering
\includegraphics[trim = 2em 2em 0em 0em , width = 0.8\textwidth]{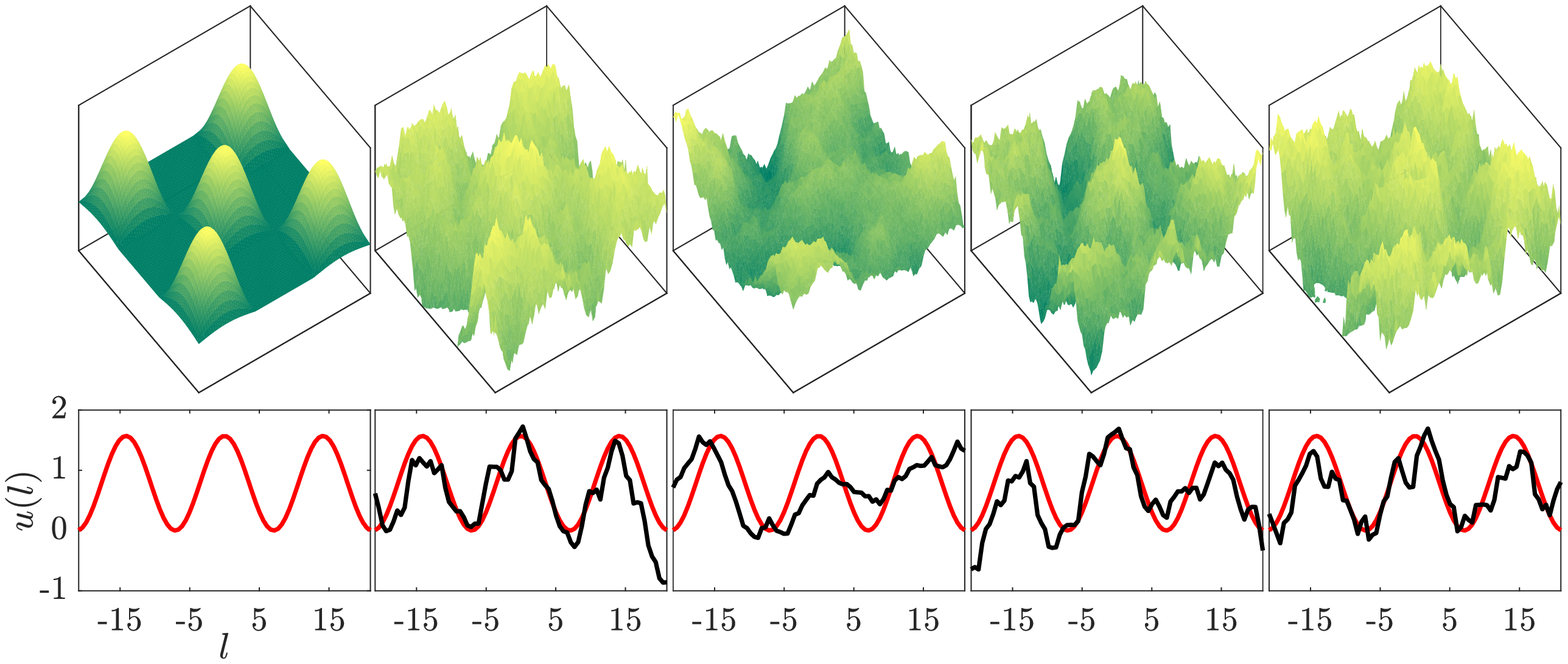}
\caption{
Top row: the true log-density function (left) and four realizations of posterior log-density functions (columns 2--5). Bottom row: corresponding log-density functions along the diagonal.}\label{fig:pet_samples}
\end{figure}

We setup three scenarios with $n = 400, 1600$, and $6400$ to test the dimension scalability of our proposed methods.
In this example, we observed that the RTO density is not as accurate as in the first example in approximating the conditional posterior density. 
This allows us to more extensively test the impact of sample sizes $N_{\rm sub}$ and $K$ on the statistical efficiency of RTO-within-Gibbs (Algorithm \ref{alg:rto_mwh}) and RTO-PM (Algorithm \ref{alg:rto_pm}), respectively. 
Before discussing the numerical experiments, we present the inversion result obtained using RTO-PM with $K = 10$ and $n=6400$ in Figures \ref{fig:pet_result} and \ref{fig:pet_samples}.
Figure \ref{fig:pet_result} plots synthetic data set and the ``true'' intensity function, the
sample median and 95\% credibility bounds of log-density function along the diagonal line, and the marginal densities and the pairwise marginal densities for these hyper-parameters.
Figure \ref{fig:pet_samples} shows the true log-density function and four realizations of posterior log-density functions. 

\subsection{Sampling performance}
We first discuss the sampling performance of RTO-within-Gibbs with $N_{\rm sub} = 1$ and $N_{\rm sub} = 5$. We compute chains of length 90,000 for the three scenarios and report the corresponding ACFs and IACTs in Figure \ref{fig:pet_gibbs}.
\begin{figure}[!htbp]
\centering
\includegraphics[trim = 0em 1.7em 0em 0.5em , width = 0.8\textwidth]{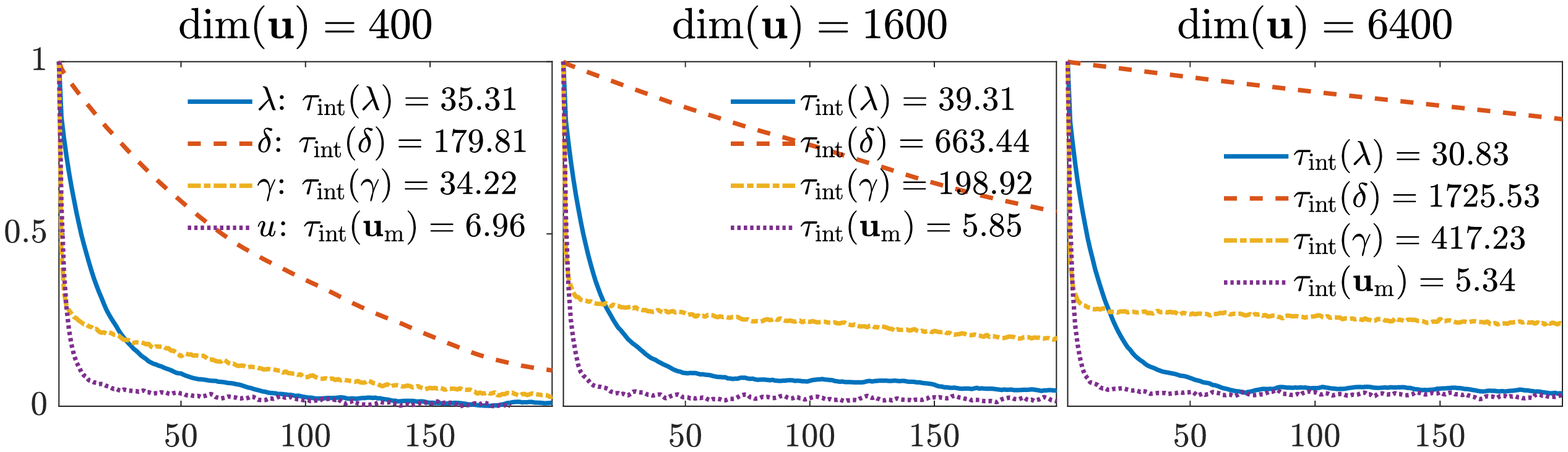}\\
\includegraphics[trim = 0em 0em 0em 3em,clip, width = 0.8\textwidth]{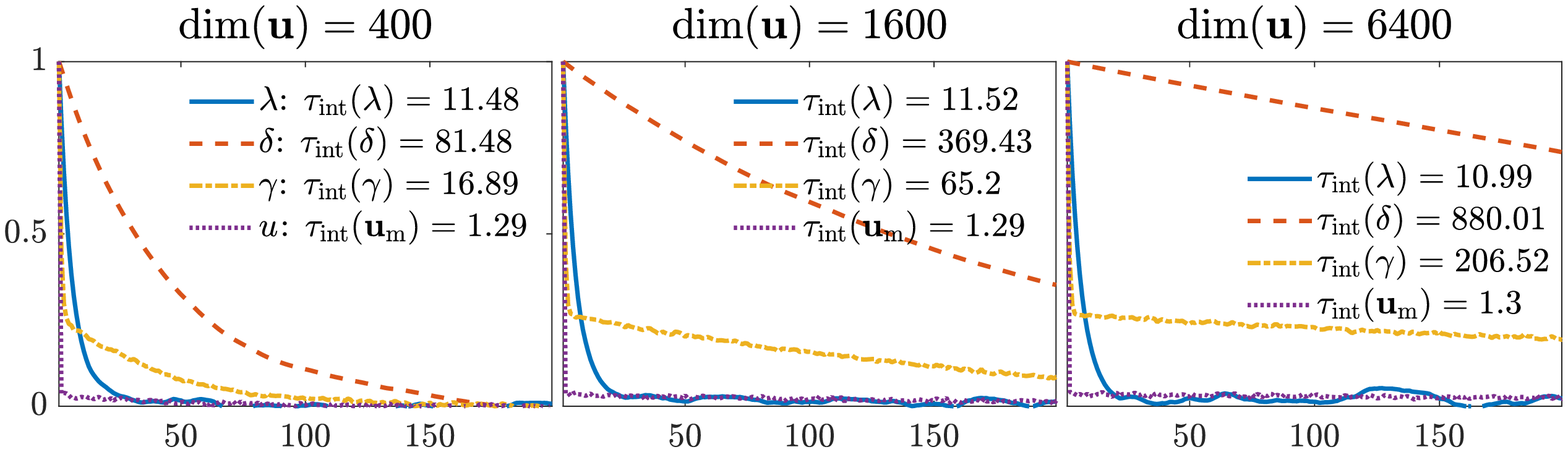}\vspace{-0.5em}
\caption{ACFs and IACTs of the hyperparameters and $\ub_{\rm m} = u(0, 0)$ computed using RTO-within-Gibbs with $N_{\rm sub} = 1$ (top row) and $N_{\rm sub} = 5$ (bottom row) for $n=400, 1600$ and $6400$.}\label{fig:pet_gibbs}
\end{figure}
In this example, RTO-within-Gibbs with $N_{\rm sub} = 1$ produces poorly mixing chains compared to the first example. 
Apart from the parameter dimensionality, a potential cause is that the chance of updating the $\ub$-chain, which is the same as the acceptance rate of RTO in Algorithm \ref{alg:rto_mwh}, is only about $49\%$ in this case. 
This can also increase the autocorrelation of other chains in the Gibbs update. 
By using $N_{\rm} = 5$, the chance of updating the $\ub$-chain is improved to about $100\%$. As a result, we observed that the IACTs of various chains produced with $N_{\rm sub} = 5$ are substantially lower than those of $N_{\rm sub} = 1$. 

However, RTO-within-Gibbs still suffers from the parameter dimensionality. Similar to the example in Section \ref{sec:elliptic}, the IACTs of the $\delta$-chains are approximately linear to the parameter dimension $n$, and the IACTs of the $\lambda$-chains and $\ub_{\rm m}$-chains do not change with the parameter dimensions. 
In contrast to the first example where the IACTs of the $\gamma$-chains do not change with $n$ (cf. Figure \ref{fig:ellptic_gibbs}), the mixing of the $\gamma$-chains in this example also depends on $n$.
Here we observed that the IACTs of the $\gamma$-chains increase with the parameter dimension $n$ and the ACFs of the $\gamma$-chains have a sharp drop initially and then exhibit a slow decay afterwards. 
The cause of this effect remains unclear.

In our second experiment, to analyze the impact of the sample size $K$ on RTO-PM, we estimated the standard deviation of the log-pseudo-marginal density, $\text{var}[\log p_{K}(\thetab|\data)]^{\frac12}$, for various sample sizes and various posterior samples of $\thetab$. 
\begin{figure}[!htbp]
\centering
\pic{.4}{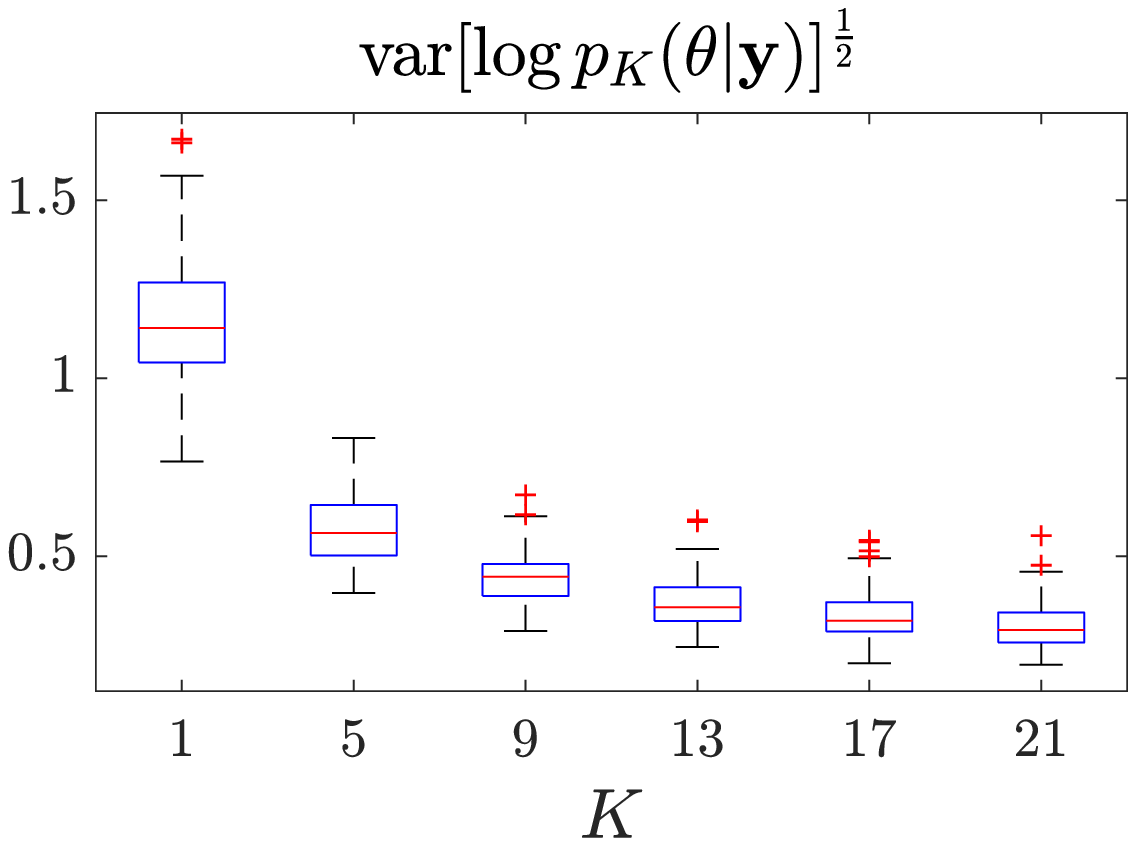}\vspace{-1em}
\caption{The box plot of $\text{var}[\log p_{K}(\thetab|\data)]^{\frac12}$ for various $K$ and posterior hyperparameters.}\label{fig:pet_var_pm}
\end{figure}
The box plot of estimated values are shown in Figure \ref{fig:pet_var_pm}.
We observed that both the spread and median of $\text{var}[\log p_{K}(\thetab|\data)]^{\frac12}$ decreases with increasing $K$.
With $K = 1$, the median of $\text{var}[\log p_{K}(\thetab|\data)]^{\frac12}$ is about $1.2$ and reduces to about $0.6$ with $K=5$. 
Using the rule-of-thumb discussed in Section \ref{remark:pm_comp}, we expect the IACTs of the hyperparameter chains produced with $K=1$ can be reduced by using a larger $K$ value in this example. 

We simulate RTO-PM with $K = 1, 5$, and $10$. The ACFs and IACTs are reported in Figure \ref{fig:pet_pm}. 
As expected, in all cases, the Markov chains generated by RTO-PM do not exhibit the dimension scalability issue, where it is shown that as $n$ increases the IACTs of the hyper-parameters stabilizes.
\begin{figure}[htbp]
\centering
\includegraphics[trim = 0em 1.7em 0em 1em , width = 0.8\textwidth]{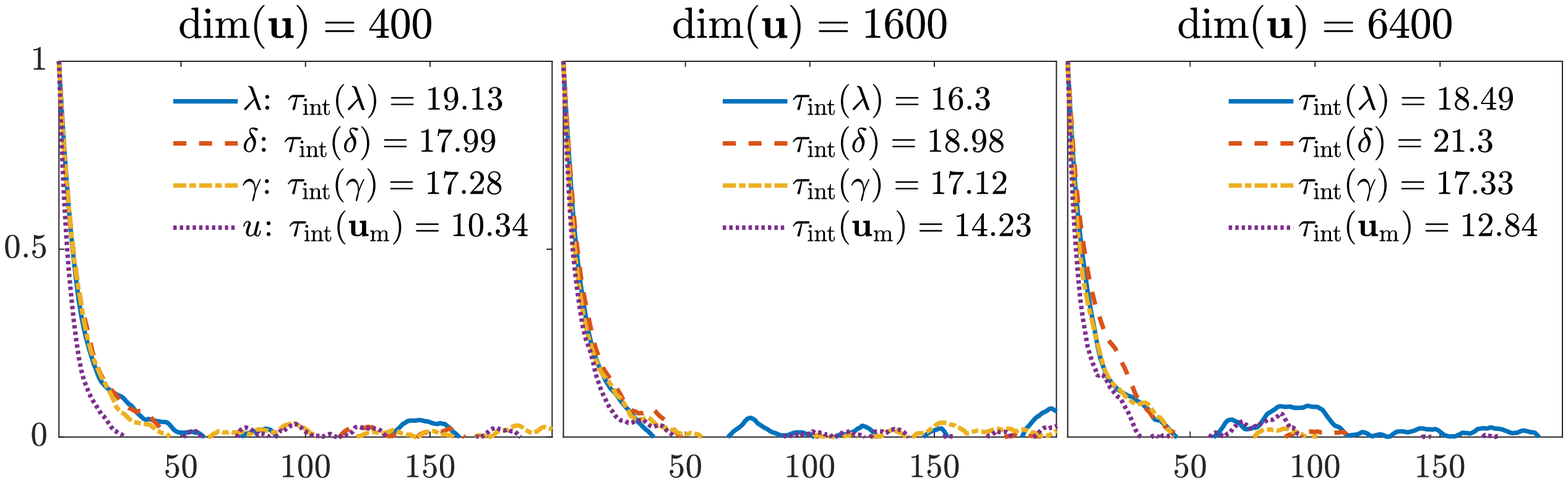}\\
\includegraphics[trim = 0em 1.7em 0em 3em,clip, width = 0.8\textwidth]{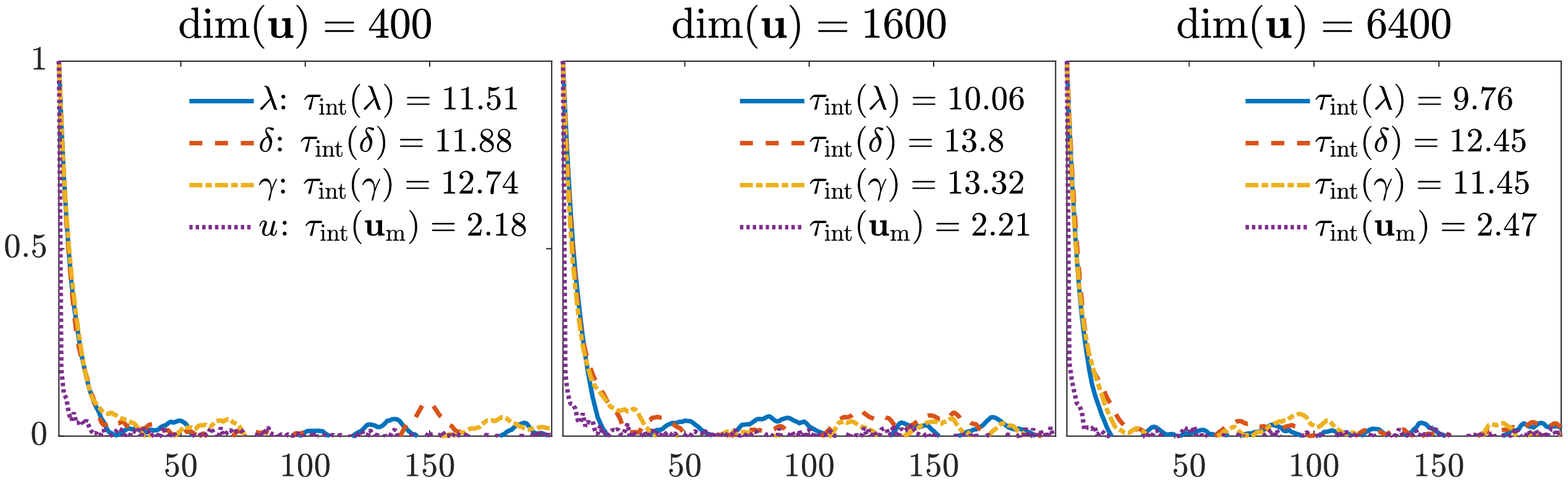}\\
\includegraphics[trim = 0em 0em 0em 3em,clip, width = 0.8\textwidth]{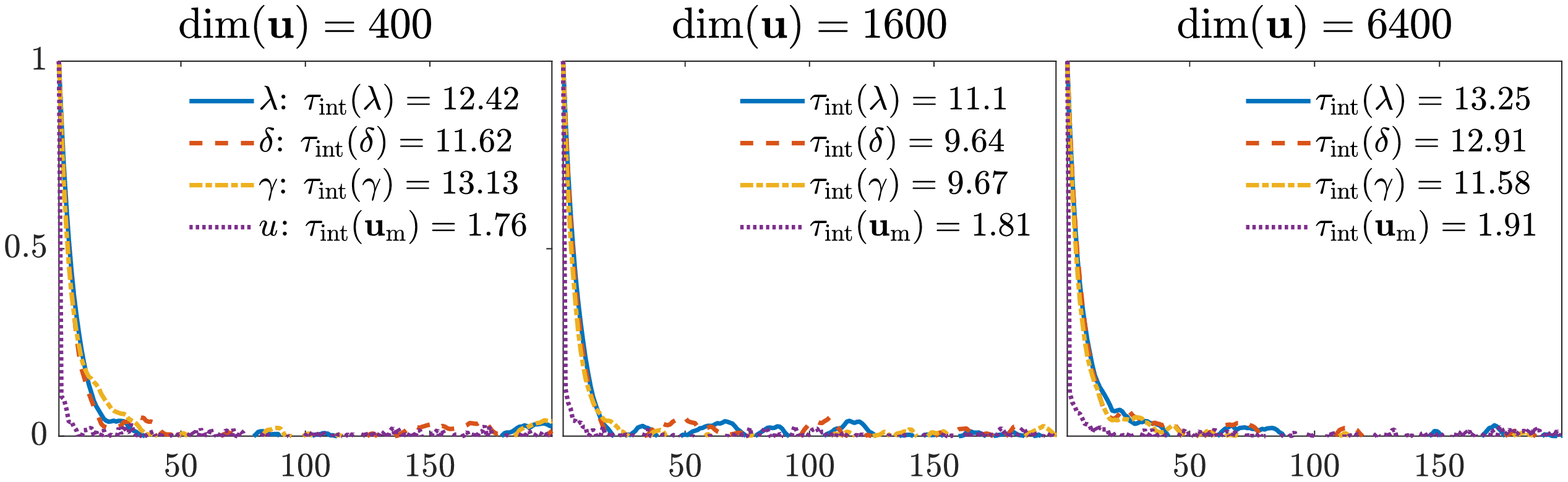}\vspace{-0.5em}
\caption{ACFs and IACTs of the hyperparameters and $\ub_{\rm m} = u(0, 0)$ computed using RTO-PM with $K=1$ (top row), $K=5$ (middle row), and $K=10$ (bottom row) for $n=400, 1600$ and $6400$.}\label{fig:pet_pm}
\end{figure}
It also confirms that, with the sample size $K=5$, both the hyperparameter chains and parameter chains have significantly smaller IACTs compared to those of $K=1$. 
Although the median of $\text{var}[\log p_{K}(\thetab|\data)]^{\frac12}$ can be further reduced to about $0.4$ with $K=10$, the IACTs of the resulting Markov chains reported in the last row of Figure \ref{fig:pet_pm} are comparable to those of $K=5$. 
This result agrees with the interpretation of the statistical efficiency of RTO-PM discussed in Section \ref{remark:pm_comp}---further accuracy improvement in the pseudo-marginal density computation will not necessarily  improve the statistical efficiency of RTO-PM.

\section{Discussion}
\label{sec:Conclusions}
In this work, we developed scalable randomize-then-optimize-based MCMC methods for solving hierarchical Bayesian inverse problems with high-dimensional model parameters, nonlinear forward models, and a broader class of hyperparameters.
In particular, we designed the RTO-within-Gibbs method, where RTO is used as a proposal within the Gibbs update, and the RTO-PM method, in which RTO is used for estimating the marginal posterior density over hyperparameters. 
In RTO-within-Gibbs, we presented an efficient Gibbs updating formula based on the inverse CDF method and offline computation.
We also extended our methods to nonlinear inverse problems with Poisson-distributed measurements.

We demonstrated the performance of our methods on numerical examples in PDE-constrained inverse problems and PET with different types of hyperparameters. 
In the PDE example, our methods are used to estimate the noise level in the likelihood, the variability of the prior, and the correlation length of the prior. In the PET example, our methods are used to estimate the unknown intensity of the radiation sources, the variability of the prior, and the correlation length of the prior.
The numerical results confirm our interpretation of the parameter-dimension scalability of these methods. RTO-within-Gibbs can be efficiently applied in some low-dimensional parameter cases, but its sampling efficiency deteriorates with the parameter dimension. In comparison, RTO-PM is robust to the parameter dimension and out-performs RTO-within-Gibbs in all test cases. 

There are many ways to extend the work described here. For example, we can apply RTO-PM to problems with high-dimensional hyperparameters or other types of priors, e.g., the non-stationary Gaussian process in \cite{roininen2019hyperpriors} and the hybrid prior in \cite{yao2016tv,zhou2020bayesian}.
We note that the efficiency of RTO-PM is sensitive to the variance of the log-marginal-density estimator, in which a number of computationally costly model solves is involved. This offers the opportunity of accelerating RTO-PM by employing surrogate models and the delayed acceptance method \cite{christen2005markov,liu1998sequential}. %
Furthermore, the existing work on the pseudo-marginal MCMC focuses on the efficiency of exploring the hyperparameters, whereas the primary quantity of interest in inverse problems is often the model parameter. Thus, it will be important to extend the analysis of \cite{andrieu2016establishing,doucet2015efficient} to characterize the efficiency of RTO-PM in exploring model parameters and the computational complexity of the associated expectation estimators.%

%% file: sec_appendix_a.tex
\section{Derivations and proofs of RTO in Section \ref{sec:RTO}}
\label{sec:app_a}

\subsection{RTO in the Gaussian likelihood case}
\label{sec:app_a1}
We derive the RTO formula in Section \ref{sec:gauss_rto} that can be applied to problems with hyperparameters by establishing its equivalence with the result of \cite{wang2019scalable}.
To be aligned with the result of \cite{wang2019scalable}, we first apply the whitening transforms
\begin{gather}\label{eq:white}
\vb = \big(\delta\prprech\big)^{\frac12}\big( \ub - \prmean \big), \quad {\rm and} \quad \zb = \big(\lambda^{-1} \obscov\big)^{-\frac12}\data, 
\end{gather}
and define the transformed forward model
\begin{gather*}
\Gb(\vb) \coloneqq \big(\lambda^{-1} \obscov\big)^{-\frac12} \forward\big(\big(\delta\prprech\big)^{-\frac12}\vb + \prmean\big) - \zb.
\end{gather*}
This defines the conditional posterior
\begin{gather*}
p(\vb | \zb, \lambda, \delta, \gamma) \propto (2\pi)^{-\frac{m+n}{2}}\,\lambda^{\frac{m}{2}}\,\det\big(\obscov\big)^{-\frac12} \exp\Big( -\frac12\big\Vert\Gb(\vb) \big\Vert^2 - \frac{1}{2}\big\Vert\vb \big\Vert^2 \Big).
\end{gather*}

Given the reference parameter $\vb_\ast = (\delta\prprech\big)^{\frac12}( \ub_\ast - \prmean )$, we consider the reduced SVD of the linearized forward model $\nabla \Gb(\vb_\ast)$:
\begin{gather}\label{eq:wsvd}
\nabla \Gb(\vb_\ast) = \Phib_{\rm L} \Sb \Phib_{\rm R}^\top ,
\end{gather}
where $\Phib_{\rm L} \in \R^{m \times r}$, $\Sb \in \R^{r \times r}$, and $\Phib_{\rm R} \in \R^{n \times r}$.
Proposition 3 in \cite{wang2019scalable} defines the RTO mapping between an i.i.d. Gaussian random variable $\xib\sim \N(0, \Ib_n)$ and the target random variable $\vb$:
\begin{gather}\label{eq:wrto_map}
\Phib_{\rm R}\,\Big[(\Sb^2 + \Ib_r)^{-\frac12} \big( \Phib_{\rm R}^\top  \vb + \Sb \Phib_{\rm L}^\top \Gb(\vb) \big)\Big] + (\Ib_n - \Phib_{\rm R}\Phib_{\rm R}^\top)\,\vb =  \xib,
\end{gather}
and shows that the resulting probability density of $\vb$ takes the form
\begin{gather}\begin{aligned}\label{eq:wrto_density}
\widetilde{p}_{\rm RTO}(\vb) & = (2\pi)^{-\frac{n}{2}}\frac{\det\big( \Ib_r + \Sb \Phib_{\rm L}^\top \nabla \Gb(\vb) \Phib_{\rm R} \big)}{\det(\Sb^2 + \Ib_r)^{\frac12} }  \\
& \quad \;\; \exp\Big( - \frac12 \big\| (\Sb^2 + \Ib_r)^{-\frac12} \big( \Phib_{\rm R}^\top \vb + \Sb \Phib_{\rm L}^\top \Gb(\vb) \big) \big\|^2 - \frac12\big\| (\Ib_n - \Phib_{\rm R} \Phib_{\rm R}^\top) \vb \big\| ^2 \Big).
\end{aligned}\end{gather}

Applying the whitening transforms in \eqref{eq:white}, the SVD of $\nabla \Gb(\vb_\ast)$ also takes the form
\begin{gather*}
\big(\lambda^{-1} \obscov\big)^{-\frac12} \Jb(\ub_\ast) \big(\delta\prprech\big)^{-\frac12} = \Phib_{\rm L} \Sb \Phib_{\rm R}^\top,
\end{gather*}
where $\Phib_{\rm L} $ and $\Phib_{\rm R} $ are matrices with orthonormal columns. Multiplying both sides of the above equation by $(\lambda^{-1} \obscov)^{-\frac12}$ on the left and $(\delta\prprech)^{-\frac12}$ on the right leads to
\begin{gather*}
\frac{\lambda}{\delta} \obscov^{-1} \Jb(\ub_\ast) \prprech^{-1} = \big(\lambda^{-1} \obscov\big)^{-\frac12} \Phib_{\rm L} \Sb \Phib_{\rm R}^\top\big(\delta\prprech\big)^{-\frac12}.
\end{gather*}
With $\Yb = (\lambda^{-1} \obscov)^{-\frac12} \Phib_{\rm L}$ and $\Xb = (\delta\prprech)^{-\frac12} \Phib_{\rm R}$, we recover the generalized SVD in \eqref{eq:gsvd}.
Then, substituting the identities $\Phib_{\rm L} = (\lambda^{-1} \obscov )^{\frac12} \Yb$, $\Phib_{\rm R} = (\delta\prprech)^{\frac12} \Xb$, and $\vb = (\delta\prprech)^{\frac12}( \ub - \prmean )$ into \eqref{eq:wrto_map} and \eqref{eq:wrto_density}, we obtain the RTO map in \eqref{eq:oblique} and the RTO density in \eqref{eq:rto_density}, respectively. 

\subsection{Proof of Proposition \ref{prop:rto_is}}
\label{sec:app_a2}
\newcommand{\mym}{\big(\forward(\param)\!-\data\big)}
\newcommand{\mymnew}{\big(\forward(\param)\!-\data\!-\gb(\ub_r)\big)}
We first express the second moment of $w(\ub; \lambda, \delta, \gamma)$ as
\begin{gather*}\begin{aligned}
\mathbb{E}_{p_{\rm RTO}}\big[ w(\ub; \lambda, \delta, \gamma)^2\big] %
& \!=\! \int \bigg( \frac{f(\ub| \data,\lambda, \delta, \gamma)}{p_{\rm RTO}(\ub| \lambda, \delta, \gamma)}\bigg)^2   p_{\rm RTO}(\param|\lambda, \delta, \gamma) d\param \\
& \! = \!\mathbb{E}_{p_0}\!\bigg[ \frac{\det\big(\Sb^2 + \Ib_r\big)^{\frac12}}{\det\big( \Ib_r \!\!+\! \Sb \Yb^\top \nabla_{\ub} \Fb\big( \ub \big) \Xb \big)} \like(\data|\param,\lambda)^2 \exp\left(- \!\frac{1}{2}\! \left\| \ub_r \right\|^2 \!+\! \frac{1}{2}\! \left\| \Theta(\ub_r; \ub_\perp) \right\|^2 \right) \!\!\bigg].\!\!\!\!\! %
\end{aligned}
\end{gather*}
Given that the mapping $\ub_r \mapsto \Theta(\ub_r; \ub_\perp)$ is invertible for all $\ub_r \in \R^r$ and $\ub_\perp \in {\rm kernel}(\Xb)$, there exists a constant $C_1 > 0$ such that
\begin{gather*}
\sup_{\ub} \frac{\det\big(\Sb^2 + \Ib_r\big)^{\frac12}}{\det\big( \Ib_r + \Sb \Yb^\top \nabla_{\ub} \Fb\big( \ub \big) \Xb \big)} = C_1 < \infty.
\end{gather*}
Following the definition of the Gaussian likelihood function in \eqref{eq:like_gauss}, we have
\begin{gather}\label{eq:bound_second}
\!\!\mathbb{E}_{p_{\rm RTO}}\!\big[ w(\ub; \lambda, \delta, \gamma)^2\big] \!\leq\! C_2\, \mathbb{E}_{p_0}\!\Big[ \exp\Big( \!\!-\! \lambda\left\Vert\forward(\param)\!-\data \right\Vert^2_{\obscov^{-1}} \!-\! \frac{1}{2} \left\| \ub_r \right\|^2 \!+\! \frac{1}{2} \left\| \Theta(\ub_r; \ub_\perp) \right\|^2 \!\!\Big) \Big],\!\! %
\end{gather}
where $C_2 = C_1 (2\pi)^{-\frac{m}{2}}\,\lambda^{\frac{m}{2}}\,\det\big(\obscov\big)^{-\frac12}$. We need to show the expectation in the right hand side of \eqref{eq:bound_second} is bounded. 
Towards this goal, we rewrite the term within the exponential function as
\begin{gather}\begin{aligned}
Q(\ub, \ub_r) \coloneqq  & - \lambda\left\Vert\forward(\param)\!-\data \right\Vert^2_{\obscov^{-1}} - \frac{1}{2} \left\| \ub_r \right\|^2 + \frac{1}{2} \left\| \Theta(\ub_r; \ub_\perp) \right\|^2 \\
=& - \frac{1}{2} \ub_r^\top \big(\Ib_r - (\Sb^2+\Ib_r)^{-1}\big) \ub_r + \mym^\top \Yb  \Sb (\Sb^2+\Ib_r)^{-\frac12} \ub_r  \\
&  - \frac{1}{2} \mym^\top \big( 2\,\lambda\obscov^{-1} - \Yb \Sb^2 (\Sb^2+\Ib_r)^{-1} \Yb^\top \big)  \mym \\
=& - \frac{1}{2}\ub_r^\top \Db_1 \ub_r - \frac{1}{2} \mym^\top \Db_2 \mym + \mym^\top \Yb  \Sb (\Sb^2+\Ib_r)^{-\frac12} \ub_r,
\end{aligned}
\end{gather}
where $\Db_1 = \Ib_r - (\Sb^2+\Ib_r)^{-1}  \in \R^{r \times r}$ and $\Db_2 = 2\,\lambda\obscov^{-1} - \Yb \Sb^2 (\Sb^2+\Ib_r)^{-1} \Yb^\top \in \R^{m \times m}$.
The matrix $\Db_1$ is positive semidefinite since 
\begin{equation}
\Db_1 = \Ib_r - (\Sb^2+\Ib_r)^{-1} = (\Sb^2+\Ib_r)(\Sb^2+\Ib_r)^{-1} - (\Sb^2+\Ib_r)^{-1}= \Sb^2 (\Sb^2+\Ib_r)^{-1}.
\end{equation} 
Extending the $\lambda^{-1}\obscov$-orthogonal basis $\Yb\in \R^{m \times r}$ into a complete $\lambda^{-1}\obscov$-orthogonal basis $\bar\Yb = [\Yb, \Yb_\perp]\in \R^{m \times m}$ and embedding the corresponding diagonal matrix $\Sb\in \R^{r \times r}$ into a diagonal matrix $\bar\Sb\in \R^{m \times m}$ such that $\bar\Sb_{ii} = \Sb_{ii}$ for $i = 1, \ldots, r$ and $\bar\Sb_{ii} = 0$ for $i = r+1, \ldots, m$, we have 
\begin{gather*}
\bar\Yb^\top (\lambda^{-1}\obscov) \bar\Yb = \Ib_m,\;\; \bar\Yb^{-\top} =  (\lambda^{-1}\obscov) \bar\Yb, \;\; {\rm and }\;\;\bar\Yb^{-1} = \bar\Yb^\top (\lambda^{-1}\obscov).
\end{gather*}
This way, the matrix $\Db_2$ can be expressed as 
\begin{gather}\begin{aligned}
\Db_2 %
& = \lambda\obscov^{-1} \big( 2\,\Ib_m - \lambda^{-1}\obscov \Yb \Sb^2 (\Sb^2+\Ib_r)^{-1} \Yb^\top\big) \\
& = \lambda\obscov^{-1} \big( 2\,\big(\bar\Yb^{-\top}\bar\Yb^{\top}\big)  - \big(\bar\Yb^{-\top}\bar\Yb^\top\big)\lambda^{-1}\obscov \bar\Yb \bar\Sb^2 (\bar\Sb^2+\Ib_m)^{-1} \bar\Yb^\top \big((\lambda^{-1}\obscov) \bar\Yb\bar\Yb^{\top}\big) \big)\\
& = \lambda\obscov^{-1} \bar\Yb^{-\top} \big( 2\,\Ib_m -  \bar\Sb^2 (\bar\Sb^2+\bar\Ib_m)^{-1} \big) \bar\Yb^{\top}\\
& = \bar\Yb \big( \Ib_m + (\bar\Sb^2+\Ib_m)^{-1}\big) \bar\Yb^\top,
\end{aligned}\end{gather}
which is positive definite. Thus, by defining a vector
\begin{gather}\begin{aligned}
\gb(\ub_r) &  = \Db_2^{-1} \Yb  \Sb (\Sb^2+\Ib_r)^{-\frac12} \ub_r \in \R^{m} = \lambda\obscov^{-1} \Yb \big( \Ib_r + (\Sb^2+\Ib_r)^{-1}\big)^{-1}\Sb (\Sb^2+\Ib_r)^{-\frac12} \ub_r,
\end{aligned}\end{gather}
we can express the function $Q(\ub, \ub_r)$ in the form of
\begin{gather*}
Q(\ub, \ub_r) = - \frac{1}{2}\ub_r^\top \Db_1 \ub_r\! - \frac{1}{2} \mymnew^\top \Db_2 \mymnew \!+ \frac12\gb(\ub_r)^\top\Db_2\gb(\ub_r).
\end{gather*}
The term $\gb(\ub_r)^\top\Db_2\gb(\ub_r)$ in the above equation takes the form
\begin{gather}
\begin{aligned}
\gb(\ub_r)^\top\Db_2\gb(\ub_r) & =  \gb(\ub_r)^\top \big(\Yb  \Sb (\Sb^2+\Ib_r)^{-\frac12} \ub_r\big) \\
& = \ub_r^\top \big( \Ib_r + (\Sb^2+\Ib_r)^{-1}\big)^{-1}  \Sb^2(\Sb^2+\,\Ib_r)^{-1}  \ub_r \\
& = \ub_r^\top \big( \Sb^2(\Sb^2+2\,\Ib_r)^{-1} \big) \ub_r.
\end{aligned}\end{gather}
Using the above identity, we have $Q(\ub, \ub_r) \leq 0$, because the function $Q(\ub, \ub_r)$ can be simplified as
\begin{gather}
Q(\ub, \ub_r) = - \frac{1}{2}\ub_r^\top \Db_3 \ub_r - \frac{1}{2} \mymnew^\top  \mymnew,
\end{gather}
where the matrix 
\begin{gather*}
\Db_3 = \Db_1 -  \Sb^2(\Sb^2+2\,\Ib_r)^{-1} = \Sb^2 \Big( (\Sb^2+\Ib_r)^{-1}\!\! - (\Sb^2+2\,\Ib_r)^{-1} \Big) = \Sb^2(\Sb^2+\Ib_r)^{-1}(\Sb^2+2\,\Ib_r)^{-1}
\end{gather*}
is positive semidefinite. %
Substituting $Q(\ub, \ub_r) \leq 0$ into the inequality \eqref{eq:bound_second}, we have
\begin{gather*}
\mathbb{E}_{p_{\rm RTO}}\big[ w(\ub; \lambda, \delta, \gamma)^2\big] \leq C_2 \Big[ \exp\Big( Q(\ub, \ub_r) \Big) \Big] \leq C_2.
\end{gather*}
Therefore, the result of Proposition \ref{prop:rto_is} follows. 

\subsection{Proof of Proposition \ref{prop:trust_region}}
\label{sec:app_a3}
Under Assumption \ref{assum:assum2}, we first show that the original mapping $\ub_r \mapsto \Theta(\ub_r; \ub_\perp)$ is locally diffeomorphic for $\ub_r \in \mathbb{S}(\varepsilon)$.
Recalling that $\nabla_{\ub_r}\!\!\Theta_{\rm R}(\ub_r; \ub_\perp) = \Sb \Yb^\top \big( \Jb(\Xb \ub_r + \ub_\perp + \prmean) - \Jb(\ub_\ast)\big) \Xb$, the determinant of the Jacobian of the mapping $ \Theta$ is
\begin{gather}\begin{aligned}
\det\big( \nabla_{\ub_r}\!\!\Theta(\ub_r; \ub_\perp) \big) & = \det\Big( \Ib_r + \Sb^2 + \nabla_{\ub_r}\!\!\Theta_{\rm R}(\ub_r; \ub_\perp)  \Big) \\
& = \det\big(\Ib_r + \Sb^2\big)\det\Big( \Ib_r + (\Ib_r + \Sb^2)^{-1} \nabla_{\ub_r}\!\!\Theta_{\rm R}(\ub_r; \ub_\perp)\Big) ,
\end{aligned}\end{gather}
as $\Ib_r + \Sb^2$ is positive definite. 
Following Assumption \ref{assum:assum2}, we have 
\begin{gather*}
\sup_{\ub_r \in \mathbb{S}(\varepsilon), \ub_\perp \in {\rm kernel}(\Xb)}  \varrho\Big((\Ib_r + \Sb^2)^{-1} \nabla_{\ub_r}\!\!\Theta_{\rm R}(\ub_r; \ub_\perp)  \Big) < 1, 
\end{gather*}
where $\varrho$ denotes the spectral radius of a matrix. 
This implies that $ \Ib_r + \nabla_{\ub_r}\!\!\Theta_{\rm R}(\ub_r; \ub_\perp) (\Ib_r + \Sb^2)^{-1} $ is invertible for all $\ub_r \in \mathbb{S}(\varepsilon), \ub_\perp \in {\rm kernel}(\Xb)$. 
Thus, $\det( \nabla_{\ub_r}\!\!\Theta(\ub_r; \ub_\perp)) \neq 0$ for all $\ub_r \in \mathbb{S}(\varepsilon)$.
Therefore, together with the continuous differentiability assumption of the forward model (see Assumption \ref{assum:assum1}), we have that $\ub_r \mapsto \Theta(\ub_r; \ub_\perp)$ is locally diffeomorphic for $\ub_r \in \mathbb{S}(\varepsilon)$.

To show the modified mapping $\ub_r \mapsto \widetilde\Theta(\ub_r; \ub_\perp)$ is diffeomorphic, we first show that the matrix
\begin{gather}\label{eq:jac_tr}
\nabla_{\ub_r}\!\!\widetilde\Theta(\ub_r; \ub_\perp)
= \big(\Ib_r + \Sb^2\big) \Big(\Ib_r + \big(\Ib_r + \Sb^2\big)^{-1}\nabla_{\zb_r}\!\! \Theta_{\rm R}(\zb_r; \ub_\perp)\,\nabla_{\ub_r}\Psib(\ub_r; \tilde\varepsilon, \tau) \Big)
\end{gather}
is invertible for all $\ub_r\in \R^r$, where $\zb_r = \Psib(\ub_r; \tilde\varepsilon, \tau)$.
Introducing a vector $\wb_r = \ub_r - \mb_r$ and denoting $r = \|\wb_r\|$, the Jacobian matrix $\nabla_{\ub_r}\Psib(\ub_r; \tilde\varepsilon, \tau) $ can be written as
\begin{gather*}
\nabla_{\ub_r}\Psib(\ub_r; \tilde\varepsilon, \tau) = \Big( \psi^\prime(r) - \frac{\psi(r)}{r} \Big) \Qb + \frac{\psi(r)}{r} \Ib_r, \quad {\rm where} \quad \Qb =  \frac{\wb_r}{r} \frac{\wb_r^\top}{r}. 
\end{gather*}
Note that $\Qb$ is a rank-$1$ orthogonal projector. For the case $\ub_r \in \mathbb{S}(\tilde\varepsilon(1-\tau))$, we simply have $\nabla_{\ub_r}\Psib(\ub_r; \tilde\varepsilon, \tau) = \Ib_r$ since $\psi^\prime(r) = 1$ and $\psi(r)/r = 1$. For the case $\ub_r \in \R^r \setminus \mathbb{S}(\tilde\varepsilon(1+\tau))$, we have 
\begin{gather*}
\nabla_{\ub_r}\Psib(\ub_r; \tilde\varepsilon, \tau) =  \frac{\tilde\varepsilon}{r} (\Ib_r - \Qb), 
\end{gather*}
since $\psi^\prime(r) = 0$, where $\tilde\varepsilon/r \leq 1$. For the case $\ub_r \in \mathbb{S}(\tilde\varepsilon(1+\tau)) \setminus \mathbb{S}(\tilde\varepsilon(1-\tau))$, we have
\begin{gather*}
\nabla_{\ub_r}\Psib(\ub_r; \tilde\varepsilon, \tau) =  \frac{\psi(r)}{r} (\Ib_r - \Qb) + \psi^\prime(r) \Qb, 
\end{gather*}
where $\psi(r)/r \in (0,1]$ and $\psi^\prime(r) \in [0,1]$. Thus, we conclude that $\nabla_{\ub_r}\Psib(\ub_r; \tilde\varepsilon, \tau)$ is symmetric positive semidefinite and all its the eigenvalues are located in the interval $[0, 1]$. 
This leads to
\begin{gather*}
\sigma_{\rm max}\Big(\big(\Ib_r + \Sb^2\big)^{-1}\nabla_{\zb_r}\!\! \Theta_{\rm R}(\zb_r; \ub_\perp)\,\nabla_{\ub_r}\Psib(\ub_r; \tilde\varepsilon, \tau)\Big) \leq \sigma_{\rm max}\Big(\big(\Ib_r + \Sb^2\big)^{-1}\nabla_{\zb_r}\!\! \Theta_{\rm R}(\zb_r; \ub_\perp)\Big).
\end{gather*}
Since $\zb_r = \Psib(\ub_r; \tilde\varepsilon, \tau) \in \mathbb{S}(\tilde\varepsilon) \subseteq \mathbb{S}(\varepsilon)$, we have
\begin{gather}\begin{aligned}
& \sup_{\ub_r \in \R^r, \ub_\perp \in {\rm kernel}(\Xb)} \varrho\Big(\big(\Ib_r + \Sb^2\big)^{-1}\nabla_{\zb_r}\!\! \Theta_{\rm R}(\zb_r; \ub_\perp)\,\nabla_{\ub_r}\Psib(\ub_r; \tilde\varepsilon, \tau)\Big) \\
 \leq & \sup_{\ub_r \in \R^r, \ub_\perp \in {\rm kernel}(\Xb)} \sigma_{\rm max}\Big(\big(\Ib_r + \Sb^2\big)^{-1}\nabla_{\zb_r}\!\! \Theta_{\rm R}(\zb_r; \ub_\perp)\,\nabla_{\ub_r}\Psib(\ub_r; \tilde\varepsilon, \tau)\Big) \\
 \leq & \sup_{\zb_r \in \mathbb{S}(\varepsilon), \ub_\perp \in {\rm kernel}(\Xb)} \sigma_{\rm max}\Big(\big(\Ib_r + \Sb^2\big)^{-1}\nabla_{\zb_r}\!\! \Theta_{\rm R}(\zb_r; \ub_\perp)\Big) < 1.
\end{aligned}\end{gather}
Thus, the matrix in \eqref{eq:jac_tr} is invertible for $\forall \ub_r \in \R^r$ and $\forall \ub_\perp \in {\rm kernel}(\Xb)$. 
By the continuous differentiability assumption of the forward model and Property ({\romannumeral 2}) of the function $\psi(r; \tilde\varepsilon, \tau)$, the mapping $\ub_r \mapsto \widetilde\Theta(\ub_r; \ub_\perp)$ is continuously differentiable for all $\ub_r \in \R^r$ and $\ub_\perp \in {\rm kernel}(\Xb)$.
Therefore, the modified mapping $\ub_r \mapsto \widetilde\Theta(\ub_r; \ub_\perp)$ is diffeomorphic.